\documentclass[10pt]{article}
\usepackage{fullpage}
\usepackage{booktabs} %
\usepackage[ruled]{algorithm2e} %
\usepackage{amsmath}
\usepackage{amssymb}
\usepackage{hyperref}
\usepackage{amsthm}
\usepackage{cleveref}
\usepackage{breqn}

  \newtheorem{theorem}{Theorem}[section]

  \newtheorem{lemma}[theorem]{Lemma}

\usepackage[numbers]{natbib}
\usepackage{tikz}
\usetikzlibrary{quotes, decorations.pathreplacing}

\SetAlFnt{\small}
\SetAlCapFnt{\small}
\SetAlCapNameFnt{\small}

\definecolor{ao(english)}{rgb}{0.0, 0.5, 0.0}
\def\ra{((1 + sqrt(2)) / 6)}
\def\rb{(1 / 6)}

\def\rl{(0.5)}
\def\ral{((1 + sqrt(2)) * \rl / (2 + 4 * \rl))}
\def\rbl{\rl / (2 + 4 * \rl)}

\def\rao{((1 + sqrt(2)) / 9)}
\def\rbo{(1 / 9)}
\def\rdo{((25 - 6 * sqrt(2) + 2 * sqrt(134 - 82 * sqrt(2))) / 63)}
\def\reo{((26 + 24 * sqrt(2) - 2 * sqrt(310 + 214 * sqrt(2))) / 63)}
\def\rfo{((44 + 18 * sqrt(2) - 4 * sqrt(74 + 22 * sqrt(2))) / 63)}

\DeclareMathOperator{\divv}{div}

\SetAlCapHSkip{0pt}
\IncMargin{-\parindent}

\setcitestyle{plainnat}

\title{Optimal Automated Market Makers: Differentiable Economics and Strong Duality}
\date{}
\author{Michael J. Curry \\ Harvard University 
\and
Zhou Fan \\ Harvard University 
\and
David C. Parkes \\ Harvard University}
\begin{document}

\maketitle

\begin{abstract}
The role of a market maker is to simultaneously offer to buy and sell quantities of goods, often a financial asset such as a share, at specified prices. An \textit{automated market maker} (AMM) is a mechanism that offers to trade according to some predetermined schedule; the best choice of this schedule depends on the market maker’s goals. The literature on the design of AMMs has mainly focused on prediction markets with the goal of information elicitation. More recent work motivated by DeFi has focused instead on the goal of profit maximization, but considering only a single type of good (traded with a numeraire), including under adverse selection (Milionis et al. 2022). Optimal market making in the presence of multiple goods, including the possibility of complex bundling behavior, is not well understood. In this paper, we show that finding an optimal market maker is dual to an optimal transport problem, with specific  geometric constraints on the transport plan in the dual. We show that optimal mechanisms for multiple goods and under adverse selection can take advantage of bundling, both improved prices for bundled purchases and sales as well as sometimes accepting payment ``in kind.’’ We present conjectures of optimal mechanisms in additional settings which show further complex behavior. From a methodological perspective, we make essential use of the tools of differentiable economics to generate conjectures of optimal mechanisms, and give a proof-of-concept for the use of such tools in guiding theoretical investigations.
\end{abstract}

\section{Introduction}

Market participants  like to buy low and sell high, and avoid doing the opposite, at least in expectation.
By doing this, a market participant can make a profit.
A typical {\em market maker}, which is an entity that plays a role in facilitating trades by other market participants,  achieves profits by charging a {\em spread}: they are willing to buy some quantity of a good at a fixed price, and sell the same quantity at a higher fixed price.
The goal of a market maker is to profit even in the face of potential {\em adverse selection}, which 
can arise when better informed traders  can selectively accept offers that are profitable to them (and therefore unprofitable to the market maker). 
Roughly speaking, the spread should be wider when there is greater
adverse selection.

At the same time, there is a much richer space of possibilities for market makers than simply adjusting 
the spread.
{\em Automated market makers (AMMs)}, as
used in prediction markets and decentralized finance (DeFi),
 frequently make use of these richer possibilities.
 For example, the {\em logarithmic-market scoring rule (LMSR)}~\cite{hanson2007logarithmic} and {\em constant-product market makers}~\cite{zhang2018formal} have the effect of offering a trader a
 continuum of trades at a continuum of prices.
From the perspective of {\em mechanism design},
 the use of  spread by a market maker is 
 similar to a {\em posted-price mechanism},
  and  posted-price mechanisms are well known to be 
 optimal for selling a single good to a single buyer.
Under reasonable models of adverse selection,  the same is true for market making: optimal mechanisms consist of a bid-ask spread, where the nature and extent of adverse selection determines the width of the spread~\cite{MilionisMyersonianFrameworkOptimal2023}. From the full space of mechanisms, the optimal mechanism
in this single good case
is therefore surprisingly simple.

But what if a market maker wants to make markets in multiple goods at once?
In this case, the full space of mechanisms might be extremely rich: in addition to offering a continuum of trades (as with the aforementioned, single-good AMMs), there could be bundling among goods and the potential to accept goods ``in kind'' in exchange for a lower payment. For example, a market maker might offer to buy good 1 for price $x$, or sell good 2 for price $y$, but simultaneously buy good 1 and sell good 2 charging  $z < y - x$.
To our knowledge, it is atypical for current market makers to bundle goods, 
and they instead simply adjust a separate spread per good (analogous to item-wise, posted prices).
 There are a few exceptions, though; e.g., in the case of certain order types on the {\em OneChronos} combinatorial stock exchange~\cite{OneChronosSite}, 
some experiments in {\em combinatorial prediction markets}~\cite{Kroer2016ArbitrageFree,Dudik2013Combinatorial}, and the portfolio trading of bonds.
 
 Many of these aforementioned  cases of bundling 
 are  motivated by outside business constraints; e.g.,
 there may be complements or substitutes between goods, or there may be practical reasons why a customer is forced to trade  bundles.
But what if all these outside concerns are removed? 
We want to understand whether
 a profit-maximizing market maker 
  in the  multiple-good case will want 
   to take advantage of the full richness of the space of mechanisms. 
(Indeed,
in classic auction design when selling multiple goods to a single buyer, optimal mechanisms may show this kind of complexity~\cite{Daskalakis2017Strong,Manelli2006Bundling,Kash2016Optimal}.)
This question motivates our work.
Under a deliberately simple model of additive and quasilinear traders, and considering an objective for the market maker that can take into account the possibility of adverse selection (following that of \citet{Milionis2022Automated}), we ask
 how much of the full space of possible mechanisms will a market maker want to use? 
We give answers to this question, 
and further believe that the computational and theoretical tools that we develop 
have broader interest.

\paragraph{Our methodological contributions}
The space of possible market making strategies is quite complex in the multiple good case,
and conjecturing the form of an optimal solution is not trivial.
To search through the space of mechanisms and develop our theoretical results and additional conjectures on optimal designs, we make use of 
{\em differentiable economics}~\cite{duetting2023optimalJACM}, which employs the tools of
 modern machine learning to search through the space of all mechanisms to optimize performance.
Through differentiable economics, we find conjectured optimal mechanisms that make use of the full richness of the space of possible mechanisms: displaying  bundling behavior across both buying and selling, and potentially offering a continuum of possible trades.
Differentiable economics has previously been used to find a optimal auctions in a few, stylized auction settings~\cite{duetting2023optimalJACM,Shen2019Automated}.
Here, we utilize differentiable economics, in particular {\em RochetNet}~\cite{duetting2023optimalJACM}, in studying and developing the theory for a new
kind of problem (optimal market making with multiple goods and adverse selection),
applying it from the start to guide our understanding and theoretical development. 
In doing so, we have found the methods of differentiable economics to be 
invaluable in guiding our exploration. 
The techniques we use appear to work quite well, and require only access to samples from the market maker's distribution of beliefs about traders. For this reason, we belief that future work can
 also pursue their use as a tool for data-driven design of market makers
 that can inform deployment.

\paragraph{Our theoretical contributions}
We provide a framework for formally proving the optimality of conjectured mechanisms that are generated through the differential economics pipeline. To do this, we treat market makers as strategyproof mechanisms, and show that the market maker design problem is dual to an optimal transport problem, extending the theory of optimal transport duality from~\cite{Daskalakis2017Strong,Kash2016Optimal} to allow for different constraints on the primal and dual, and giving a geometric interpretation of these constraints that allows solving the dual problem. 

In particular, we use differentiable economics to search for profitable mechanisms across multiple valuation distributions.
We find that some of the learned mechanisms make use of {\em mixed bundling}; i.e., trading in goods separately at one price, but also offering bundle discounts or taking payments in kind (accepting a sale of one good to discount the purchase of another). Other mechanisms offer a continuum of allocations in certain regions.
Based on this, 
we conjecture that the full space of mechanisms is in fact potentially useful to the market maker. 

For one of these cases, we apply the duality framework, constructing explicit dual certificates for a particular choice of distribution under a linear model of adverse selection, showing optimality for a family of mechanisms parameterized by the strength of adverse selection.
As conjectured, the optimal mechanisms offer a bid-ask spread for each good. However, they also offer bundled trades at better prices than the separate bids and asks. For a specific example in one case, the trader can buy one good for $\frac{5}{6}$, or sell another for $\frac{1}{6}$, but simultaneous buying and selling costs only $\frac{4-\sqrt{2}}{6}$. As the strength of adverse selection decreases, the prices for all  goods and bundles improve nonlinearly.

We thus answer our main question about the optimal design of market makers in this multiple good setting   in the affirmative: under a reasonable model of adverse selection, it may be optimal for a profit-maximizing market maker to engage in mixed bundling, with profits up to 11.4\% greater than separately making a market for
each good.
Beyond being of theoretical interest, we believe that this positive answer to the question 
about the richness of market making possibilities in this setting 
can motivate the use of multiple-good CFMMs in DeFi and perhaps also in prediction markets, as well as providing a  motivation for the increased use of combinatorial bundle order types in securities markets.

\subsection{Related Work}

\subsubsection{The multiple-good monopolist problem}

\citet{Adams1976Commodity} observed that for a monopolist selling more than one good, it may be advantageous to engage in mixed bundling, selling a bundle of multiple goods at a discounted price, while possibly also offering each good separately at an independent price.
Surprisingly, this advantage does not require that goods be complements or substitutes; rather, 
it can happen even if the goods are completely unrelated, for example with a buyer 
with additive values.

A long line of work \cite{McAfee1989Multiproduct, Mcafee1988Multidimensional,rochet1998ironing} seeks to characterize in what circumstances taking advantage of bundling is optimal.
\citet{Pavlov2011Optimal}, \citet{Manelli2006Bundling}, and 
\citet{Giannakopoulos2014Duality} all give explicit optimal solutions to particular special cases of the multiple-good monopolist problem, and show that these solutions  do sometimes 
 involve mixed bundling.

 \citet{Daskalakis2013Mechanism,Daskalakis2017Strong} present a general theory of strong duality, in which the optimal mechanism design problem is shown to be dual to a relaxed optimal transport problem.
They  show that depending on the particular problem, selling separately, selling only a bundle, mixed bundling, or even offering an infinite continuum of lottery bundles could be optimal.
\citet{Kleiner2019Strong} give an equivalent formulation of the problem in terms of infinite-dimensional linear programming, and \citet{Kash2016Optimal} allow for more general constraints on allocations.

\subsubsection{Prediction markets and automated market makers}

There is a rich literature on the use of AMMs to facilitate prediction markets~\cite{frongillo2017axiomatic,othman2011liquidity,Abernethy2012Characterization,abernethy2013efficient}. There are direct connections between truthful mechanisms and proper scoring rules~\cite[e.g.]{Frongillo2014General}. 

However, the goals of prediction markets are very different than those in the present paper:
 the purpose of a prediction market is  to elicit  useful information from traders.  Profit is usually not a concern, 
 and the market maker may be happy to tolerate a bounded loss if this helps the  market run smoothly. Because the goal is to elicit accurate predictions, these mechanisms are also  designed with  properties that are almost diametrically opposed to our goals. For example, they should be not merely weakly truthful, but strictly truthful, to uniquely distinguish all types (in the language of scoring rules, they should be \textit{strictly} proper and not just proper). This necessarily means that there cannot be a ``no-trade'' region where many types all receive the identical zero allocation; in contrast, we will see that  optimal market makers  often must do this in our setting. 
Moreover,   most prediction markets consider a single security at once  (see \cite{Dudik2013Combinatorial,Kroer2016ArbitrageFree} for examples of combinatorial markets).
We do note that there is other work from this literature that does focus on profit~\cite{abernethy2013adaptive,othman2011liquidity,das2008adapting}. 

\subsubsection{Constant-function market makers}

Constant-function market makers (CFMM) are a class of AMMs
that start with an initial portfolio of assets (one of which may be considered the numeraire) from a 
liquidity provider (LP).
They are willing to exchange any basket of assets for another basket of assets as long as some constant function of the portfolio (referred to as the ``invariant'' or ``curve'') is preserved.
\citet{Angeris2021Replicating} describes how to construct a CFMM curve that results in a given  payoff for the liquidity provider, as long as that payoff is concave in the price.
\citet{Milionis2022Automated} give one criterion, {\em loss-versus-rebalancing}, for measuring 
the performance over time of AMMs. \citet{MilionisMyersonianFrameworkOptimal2023} consider a flexible model of adverse selection over a single time step and show that optimal CFMM ``curves''
take the form of a bid-ask spread.
\citet{Goyal2023Finding}, similar to our work, consider optimizing CFMMs for various performance goals given market maker beliefs about the value of goods. Their main focus is  an objective related to liquidity, where they show a correspondence between market maker beliefs and various existing and new CFMM curves. 
They also consider optimizing for profit and loss-versus-rebalancing.
Crucially, they focus on a single-parameter setting, with one token trading against one numeraire.

\subsubsection{Characterization of truthful mechanisms}

\citet{Rochet1987necessary} gives an extremely useful characterization of dominant-strategy incentive compatible (DSIC, or strategyproof) mechanisms. Essentially, every truthful mechanism can be identified with a convex function mapping the participant's true type to their utility for bidding truthfully. The allocation rule (which uniquely determines the payment rule) for a given type corresponds to the (sub)gradient of this function. Finally, one can interpret the convex conjugate of this function \cite{Rockafellar2015Convex} as a summary of the ``menu'' mapping possible allocations to prices.
\citet{Frongillo2014General} extend this general idea to additional mechanism design problems such as truthful elicitation via scoring rules and liquidity provision in prediction markets.

There are in fact further similarities between all of the above results, which identify each mechanism with the utility it induces for an honest agent, %
 and connect strategyproofness to convexity. For example, \cite{Frongillo2023Axiomatic} show an equivalence between CFMMs and scoring rules. \citet{Angeris2021Replicating} characterize CFMMs in terms of a concave ``value function'' which, up to a linear translation, corresponds to the utility function of the mechanism participant. \citet{Frongillo2014General}, Section 1.2, give an exhaustive survey of related work with similar results.

\subsubsection{Automated mechanism design and differentiable economics}

The general lack of theoretical progress on the problem of optimal mechanism design has motivated the use of \textit{automated mechanism design}: the use of computational techniques to search for high-performing mechanisms on specific problem instances~\cite{Conitzer02Mechanism,Sandholm2003Automated,Likhodedov2004Methods,Sandholm2015Automated}. In the typical model of mechanism design, the valuation distribution is common knowledge,
and 
if samples are available from this distribution,  it is natural to formulate the mechanism design problem as a machine learning problem; indeed, a number of works have considered the learning-theoretic properties of various classes of mechanisms~\cite{Balcan2016Sample,Balcan2019Estimating,Balcan2018General}.

More recently, \citet{duetting2023optimalJACM} %
introduce \textit{differentiable economics}: the use of rich, flexible function approximators, trained via stochastic gradient descent, for strategyproof auction design. One of their architectures, {\em RegretNet}, works for multi-bidder auctions, with strategyproofness constraints  approximately enforced via a penalty term. Another architecture, {\em RochetNet}, is always strategyproof but  only well-defined for single-bidder auctions.
This single bidder fits with the present setting, where we seek the optimal design 
given knowledge of a distribution on the possible types of a 
single trader who will interact
with the market maker, and we make use of the RochetNet approach in the present paper.
Separately,~\citet{Shen2019Automated} consider {\em MenuNet}, an architecture which takes a similar approach to RochetNet and is also limited to single-bidder settings, and to which they apply to a broader class of value models beyond additive and unit-demand.
\citet{curry2023differentiable} present a direct generalization of these two approaches, appealing to affine maximizers to give an architecture that is always strategyproof and well defined for many-bidder auctions (although it may not always be able to represent the optimal mechanism). 
There are followups improving on the architecture of RegretNet~\cite{Likhodedov2004Methods,Rahme2020Permutation,Rahme2021Auction, Curry2020Certifying} and the \citet{curry2023differentiable} approach~\cite{Duan2023Scalable}.

\section{Our model}

Here we describe our model of optimal automated market-maker design.
We first describe our family of objectives, which are a generalization of those of \citet{MilionisMyersonianFrameworkOptimal2023} to multiple goods.
We then describe our constraints, that is, the set of mechanisms we are willing to consider. These correspond  to either strategyproof and individually rational (defined below)   direct-revelation mechanisms with feasible allocation rules, 
 or menus of (lottery) bundles at associated prices.\footnote{There are also connections to constant-function market makers, which we discuss in \Cref{app:cfmm}.}
An important detail of the market setting
is that we also need to handle adverse selection, reflecting that the trader
 may be better informed than the market maker.

\subsection{Mechanism Design Objectives}
\label{sec:objectives}

We assume that the market maker has some initial belief vector $c\in \mathcal{X} = [0,b]^d$, for some maximum value $b>0$,
 on the value for each of $d\geq 1$ goods. 
A trader has their own value vector $x\in \mathcal{X} = [0,b]^d$,
 and seeks to trade  with the market maker. The trader's value is, in mechanism design terms, their \textit{type}, drawn from some distribution with probability density function, $f(x)$.

We assume that the market maker is willing to buy or sell up to one unit of each good. 
After a trade, the market maker receives a profit equal to their gain or loss of money plus the value of goods bought or sold, also considering how their beliefs may be affected by the trade. 
In particular, there is the possibility of {\em adverse selection}: the market maker may be incorrect about the value $c$, and the trader may be better informed. To capture this possibility, following \citet{Milionis2022Automated}, we  adopt an \textit{update function}, $\pi(c, x): \mathcal{X} \times \mathcal{X} \rightarrow \mathcal{X}$.
 When the market maker observes a trader with value $x$, they update their own belief about the value of the goods from $c$ to $\pi(c, x)$ (note that, as in \citet{Milionis2022Automated}, we require that $\pi(c,c)=c$).
The resulting profit is the revenue minus the loss due to the sale of goods (respectively, payment to trader and profit from gain of goods). 

Suppose that after a trade takes place, the trader's  allocation of goods is $a \in [-1,1]^d$ and their  payment is $p \in \mathbb{R}$. Note that the trader may gain or lose a fractional amount of each good, and that payment $p$ may be negative in the case that they receive a payment from the market maker. 
Then the \textit{trader's utility} is their value of the allocation (perhaps negative, in the case of a sale)
 minus their payment: $a \cdot x - p$.
The \textit{market maker's profit} is  the revenue from the sale (perhaps negative if they pay the trader), minus the loss or gain from the transfer of goods, this measured at their updated belief about the value of the goods: $p - a \cdot \pi(c, x)$.
In other words, given beliefs about the value of goods,  the market maker and trader each have 
 \textit{additive} and \textit{quasilinear} utility.

We highlight three important special cases, each with a different definition of belief update, $\pi$:
\begin{enumerate}
\item When $\pi(c, x) = c$, this is the ``noise trading'' case, where the mechanism designer believes that the bidders are simply uninformed, with no adverse selection.
\item In the above case, if we further have that $c=0$, this recovers the ``selling-only'' case studied in the auction setting of \citet{Daskalakis2017Strong}.
\item If $\pi(c, x) = \lambda c + (1 - \lambda)x$, for some $0 \le \lambda \le 1$, this is the ``linear interpolation'' case from \citet{MilionisMyersonianFrameworkOptimal2023}. This is a simple model of adverse selection.
\end{enumerate}

\subsection{Class of mechanisms}

Following \citet{Rochet1987necessary}, we can consider three equivalent perspectives on our mechanisms. 
The first perspective is that of a \textit{strategyproof}, direct-revelation mechanism, in which the trader  directly reports their value vector $x$, and allocations and payments are determined according to an \textit{allocation rule}, $a(x) : \mathcal{X} \rightarrow [-1,1]^d$ and {\em payment rule}, $p(x): \mathcal{X} \rightarrow \mathbb{R}$.
The bidder then has a utility, which is a function of their valuation: $u(x) = a(x) \cdot x - p(x)$.
Strategyproofness means that there should be no incentive for the bidder to report anything other than their true value: that is, $\forall x, x': a(x') \cdot x - p(x') \leq u(x)$. 

The second perspective
is that a mechanism corresponds to offering a {\em menu}, specifying a price for any possible allocation $a$ (with infinite, positive prices for infeasible allocations), where the trader chooses freely among the menu elements.
One can denote this menu $u^*(a)$, because it is the Fenchel conjugate \cite{Rockafellar2015Convex} of the utility function $u(x)$ enjoyed by a bidder who chooses their best menu element,
 or who participates truthfully in the direct revelation mechanism.
A third perspective is that given a utility function $u(x)$, the gradients of the utility, $\nabla u(x)$,  define the direct revelation allocation rule $a(x)$, with payments $p(x) = \nabla u(x) \cdot x - u(x)$.
Rochet showed that all of these perspectives are in some sense equivalent.

We  find it especially useful to focus on the third perspective, and to design mechanisms through
the design of the trader's utility function $u(x)$,
because different mechanism design desiderata can be
easily connected to constraints on the agent's utility function.
The first three desiderata are standard in mechanism design and the fourth is special to this domain. 
\begin{itemize}
    \item \textbf{Strategyproofness} can be identified with \textbf{convexity} of the utility function~\cite{Rochet1987necessary}.
    \item \textbf{Individual rationality} means that if traders participate truthfully, they are guaranteed nonnegative utility, so can always safely participate in the mechanism; this can be identified with \textbf{nonnegativity} of the utility function: $\forall x, u(x) \geq 0$.  .
    \item \textbf{Feasibility of allocations} involves ensuring that the \textbf{gradients are bounded}. In particular, we would like that $\lVert a(x) \rVert_\infty = \lVert \nabla u(x) \rVert_\infty \leq 1$, or equivalently that $u(x)$ is \textbf{1-Lipschitz in the $\ell_1$ norm} (dual to the $\ell_\infty$ norm).
    \item In addition, there is always a \textbf{no-trade region} in the present setting,
 that includes at least the point where the trader and market maker agree on the price (and thus do not want to trade).
 Under our model, this  happens at the initial price belief $c$.
 Thus we also require  \textbf{u(c) = 0}, which by nonnegativity means utility attains its minimum at $c$.
\end{itemize}

Based on these considerations, we  define the following {\em convex set of feasible mechanisms}:
\begin{equation*}
    \mathcal{U} = \{ u : \mathcal{X} \rightarrow \mathbb{R}^+ \,\, |\,\, u \text{ convex }, \lVert u(x) \rVert_\infty \leq 1, u(c) = 0 \}.
\end{equation*}

We can also rewrite the market-maker's expected profit in terms of $u(x)$.
Since $a(x) = \nabla u(x)$, profit is $\nabla u(x) \cdot x - u(x) - \nabla u(x) \cdot \pi(c, x) = \nabla u(x) \cdot (x - \pi(c, x)) - u(x)$.
Thus, the market maker's objective of expected profit under belief updating,  for traders distributed according to $x \sim f(x)$, is:
\begin{equation}
\int_{\mathcal{X}} \left(\nabla u(x) \cdot (x - \pi(c, x)) - u(x) \right)f(x)\,dx.
    \label{eq:mechobjectiverepeat}
\end{equation}
\section{Establishing Duality}

To find the profit-maximizing mechanism, we now face an optimization problem over the space of feasible utility functions:
\begin{equation*}
\sup_{u \in \mathcal{U}} \int_{\mathcal{X}} \left(\nabla u(x) \cdot (x - \pi(c, x)) - u(x) \right)f(x)\,dx.
\end{equation*}
Our goal in this section is to establish {\em strong duality}: we will provide a dual minimization problem, and show that the primal and dual objectives coincide at the optimal point.
This will allow us to check optimality of a proposed mechanism by showing a matching dual solution.

\subsection{Linearizing the objective}
\label{ssec:linearize}

Similar to other work, \cite{Mcafee1988Multidimensional,Daskalakis2017Strong,Kash2016Optimal,rochet1998ironing},
 we can linearize our problem by doing integration by parts, establishing the following lemma.
\begin{lemma}
\label{lemma:byparts}
\begin{equation}
\int \left(\nabla u(x) \cdot (x - \pi(c,x)) - u(x)\right)f(x)\,dx = \int u \,d\mu^+ - \int u\,d\mu^-,
\end{equation}
where $\mu = \mu^+ - \mu^-$ is a signed measure which integrates to 0, defined as
\begin{equation}
\begin{aligned}
        \mu(A) &= \int_{\partial X} \mathbb{I}_A(x) f(x)(x - \pi(c, x))\cdot \hat{n}\,dx \\
        &- \int_{\mathcal{X}} \mathbb{I}_A(x) (\nabla f(x) \cdot (x - \pi(c,x)) 
         + (n+1 - \divv \pi(c,x))f(x))\,dx + \mathbb{I}_A(c).
        \end{aligned}
   \label{eq:transformedrepeat}
\end{equation}
\end{lemma}
The calculation, which is mainly an application of integration by parts, is in \Cref{app:linearize}.

\paragraph{Interpretation of $\mu$} There are three terms in $\mu(A)$.
 The first term involves positive mass distributed along the boundary of the type space. The second  term involves a density of negative mass distributed across the interior (and possibly also touching the boundaries). The third term is an indicator function denoting one positive unit of mass at $c$. Because we require that $u(c) = 0$, we can add this term without altering the result, which is useful as it serves to balance positive and negative mass so that we have a true optimal transport problem between $\mu^+$ and $\mu^-$. We will see in \Cref{sec:warmup1d,sec:optimalmulti} below that the positive mass $\mathbb{I}_A(c)$ associated with this third term
will be spread outward to exactly cover the negative mass in the ``no-trade'' regions that show up in the center of the type space.

\subsection{Statement of duality}

Following the above discussion, the mechanism design problem is to design a profit-maximizing, feasible 
utility function to the trader, $u(x)$. We define a matching dual problem, and establish strong duality.
\begin{theorem}[Strong duality, following \cite{Villani2003Topics,Daskalakis2017Strong,Kash2016Optimal}]
We have
\begin{equation}
       \sup_{u \in \mathcal{U}} \int u \,d\mu^+ - \int u\,d\mu^-  = \inf_{\gamma: \gamma_1 \succeq \mu^+, \gamma_2 \preceq \mu^-} \int \lVert x - y \rVert_1,d\gamma(x,y),
\label{eq:strongduality}
\end{equation} 
where $\succeq$ is an integral stochastic order~\cite{muller1997stochastic}, representing the relation $\alpha \succeq \beta \iff \forall f \in \mathcal{U}:\int f\,d\alpha \geq \int f\,d\beta.$
\label{thm:strongduality}
\end{theorem}
The proof is in \Cref{ssec:dualproof}. While we focus on \cref{eq:mechobjectiverepeat} as the objective in this paper, the strong duality result also applies to any other objective that can be linearized into an integral of
the trader's utility function, $u$, against some balanced $\mu^+$ and $\mu^-$.

\subsection{Interpretation of the dual constraints}

The relation which we denote here as $\succeq$ has a geometric interpretation.
Given some $\alpha \succeq \beta$, any modification to transform $\alpha$ into $\alpha'$ that is guaranteed to weakly increase the value of $\int f\,d\alpha$ will ensure that $\alpha' \succeq \beta$ as well.
In our case, $f$ is required to be convex, so any mean-preserving spread (as in the sweeping/\textit{balayage} discussed in \citet{rochet1998ironing}) of positive mass will weakly increase the value of the integral.
Additionally, we require that $f$ attain its minimum at $c$. This along with convexity means that moving positive mass in any direction coordinate-wise away from $c$ weakly increases the value of the integral.
(For negative mass, we have mean-preserving contractions and coordinate-wise moves towards $c$ respectively.)

Therefore, either mean-preserving spreads, or moving mass in any direction coordinate-wise away from $c$, can be used to improve the dual objective without paying a transport cost.
In the special case when $c$ is fixed to 0,  we recover the selling-only, multiple good monopolist problem. Mean-preserving spreads are still allowed, and moving mass away from $c=0$ amounts to moving mass upwards/rightwards, so $\succeq$ becomes precisely the convex dominance condition defined in \cite{Daskalakis2017Strong}.

Also,  we have $\alpha \succeq \alpha$, and so
 the $\succeq$ relation represents a \textit{relaxation} of the typical equality constraints in an optimal transport problem.
This relaxation of the dual is needed to maintain zero duality gap, because the primal mechanism design problem is tightened by additional constraints on the functions as a result of enforcing mechanism design desiderata such as strategyproofness.

In addition, scaling $f$ by a constant does not change anything. As such, the same $\succeq$ is defined if we instead consider $\mathcal{U}^\circ$, the cone generated by $\mathcal{U}$. $u \in \mathcal{U}^\circ$ are still convex and attain their minimum of $0$ at $c$, but are no longer necessarily 1-Lipschitz.

\subsection{Proving duality}
\label{ssec:dualproof}

\citet{Daskalakis2017Strong} generalize the standard proof of optimal transport duality from \citet{Villani2003Topics}; \citet{Kash2016Optimal} in turn builds further on each of these theories.\footnote{\citet{Kleiner2019Strong} give an alternate proof in terms of infinite-dimensional linear programming. We focus on the optimal transport interpretation because thinking geometrically in terms of transport plans is useful for actually finding dual solutions. But the linear programming interpretation may be helpful for 
some readers to keep in mind.}
We follow in this tradition to prove \Cref{thm:strongduality}, for completeness reproducing much of the same structure of the proof. The key differences in our work are that we consider a different set of feasible functions $\mathcal{U}$, and therefore a different definition of the relation $\succeq$, which requires some technical modifications to make the proof go through. Additionally, actually finding dual solutions requires a geometric interpretation of the relation $\succeq$, which is  different than the integral stochastic orders in prior work.

\subsubsection{Weak Duality}
It is easy to establish weak duality, in the same way as \cite{Daskalakis2017Strong,Kash2016Optimal}:
\begin{align*}
   \int_{\mathcal{X}} u \,d\mu^+ - \int_{\mathcal{X}} u \,d\mu^-  &\leq   \int_{\mathcal{X}}  u\,d(\gamma_1 - \gamma_2) \quad \text{ by definition of} \succeq\\
   &=  \int_{\mathcal{X} \times \mathcal{X}} (u(x) - u(y))\,d\gamma(x,y) \\
   &\leq \int_{\mathcal{X} \times \mathcal{X}} \lVert x - y \rVert\,d\gamma(x,y) \quad \text{by 1-Lipschitz property of $u$}.
\end{align*}
Weak duality already suffices for using the framework to prove optimality.
\subsubsection{Strong Duality}
We start by showing that the transport problem is dual to a an optimization problem over a pair of functions.
\begin{lemma}
     $$\inf_{\gamma: \gamma_1 \succeq \mu^+, \gamma_2 \preceq \mu^-} \int \lVert x - y \rVert_1\,d\gamma(x,y) = \sup_{\phi(x) - \psi(y) \leq \lVert x - y \rVert_1, \phi, \psi \in \mathcal{U}^\circ} \int \phi\,d\mu^+ - \int \psi\,d\mu^-.$$
\label{lemma:twofunctions}
\end{lemma}
The proof, similar to those in previous work, is deferred to \Cref{app:linearize}.

\paragraph{Optimal Solution Involves a Single Function in $\mathcal{U}$}

The next step is to show that some optimal, feasible pair $\phi, \psi$ which lie in the cone $\mathcal{U}^\circ$ can be replaced with a weakly better pair of two identical functions $\overline{\phi}, \overline{\phi}$ in $\mathcal{U}$---in other words, a single feasible mechanism. The technique here is similar but not identical to prior work, due to the need to preserve the $u(c) = 0$ property.
\begin{lemma}
  $$\sup_{\phi(x) - \psi(y) \leq \lVert x - y \rVert_1, \phi, \psi \in \mathcal{U}^\circ} \int \phi\,d\mu^+ - \int \psi\,d\mu^- = \sup_{u \in \mathcal{U}} \int u \,d\mu^+ - \int u\,d\mu^-.$$
    \label{lemma:twotoone}
\end{lemma}

\begin{proof}
First, let $\overline{\phi}(x) = \inf_y \psi(y) + \lVert x - y \rVert_1$.
$\overline{\phi}$ is still a feasible function:
\begin{itemize}
    \item it maintains convexity because $\psi(y) + \lVert x - y \rVert_1$, a sum of convex functions, is convex in $x$ and $y$~\cite{Rockafellar2015Convex,Kleiner2019Strong}.

\item $\overline{\phi}(c) = 0$ because $\psi(c) = 0$ and the $\lVert \cdot \rVert_1$ term is minimized to 0 by setting $y = c$ also.
\item Moreover, for any $y$ that is the minimizer for $\overline{\phi}(x)$, and for all $x'$, $\overline{\phi}(x') - \overline{\phi}(x) \leq \psi(y) + \lVert y - x' \rVert_1 - \psi(y) - \lVert y - x \rVert_1 \leq \lVert x' - x\rVert$, so $\overline{\phi}$ is also 1-Lipschitz in the $\ell_1$ norm. 
\end{itemize}

$\overline{\phi}$ is thus feasible, lying not just in $\mathcal{U}^\circ$ but actually in $\mathcal{U}$, and, because $\phi(x) \leq \psi(y) + \lVert x - y \rVert_1$, replacing $\phi$ by $\overline{\phi}$ weakly improves the objective (by increasing the $\mu^+$ term). So we now have a new optimal pair $\overline{\phi}, \psi$, with $\overline{\phi} \in \mathcal{U}$ and $\psi \in \mathcal{U}^\circ$.

We can now define $\overline{\psi}(y) = \sup_x \overline{\phi}(x) - \lVert x - y \rVert_1$.
$\overline{\psi}(y) \geq \overline{\phi}(x) - \lVert x - y \rVert_1$, so feasibility is maintained and the objective is weakly improved (by decreasing the $\mu^-$ term).
And in fact, $\overline{\psi} = \overline{\phi}$.
First, $\overline{\psi}(y) \geq \overline{\phi}(y) - \lVert y - y \rVert$ due to the supremum.
Then, 
\begin{align*}
    \overline{\psi}(y) &= \sup_x \overline{\phi}(x) - \lVert x - y \rVert \\
    &= \overline{\phi}(y) + \sup_x \left(\overline{\phi}(x) - \lVert x - y \rVert\right)- \overline{\phi}(y) \\
    &= \overline{\phi}(y) + \sup_x \left(\overline{\phi}(x) -  \overline{\phi}(y) - \lVert x - y \rVert\right) \\
    &\leq \overline{\phi}(y),
\end{align*}
where the last inequality follows because $\overline{\phi}$ is 1-Lipschitz, so the objective of the supremum is always nonpositive.
Thus $\overline{\psi} = \overline{\phi}$, and we have a feasible pair $\overline{\phi}, \overline{\phi}$, with $\overline{\phi} \in \mathcal{U}$, which has weakly improved the objective.
\end{proof}

Combining the above steps, we have
\begin{proof}[Proof of \Cref{thm:strongduality}]
\begin{align*}
    \inf_{\gamma: \gamma_1 \succeq \mu^+, \gamma_2 \preceq \mu^-} \int \lVert x - y \rVert_1\,d\gamma(x,y) &=  \sup_{\phi(x) - \psi(y) \leq \lVert x - y \rVert_1, \phi, \psi \in \mathcal{U}^\circ} \int \phi\,d\mu^+ - \int \psi\,d\mu^- \text{ by \Cref{lemma:twofunctions} }\\
    &= \sup_{u \in \mathcal{U}} \int u \,d\mu^+ - \int u\,d\mu^- \text{ by \Cref{lemma:twotoone}.} \qedhere
\end{align*}
\end{proof}
\section{Warmup: a 1D example}
\label{sec:warmup1d}
\begin{figure}
    \centering
    \includegraphics[width=0.4\textwidth]{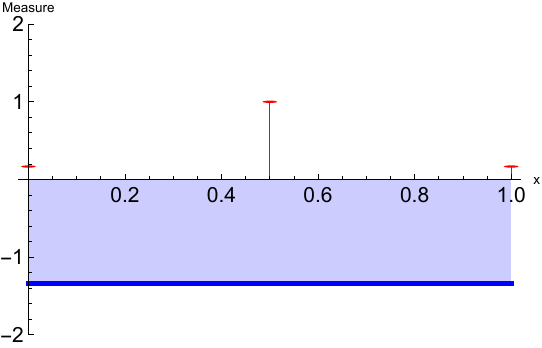}\includegraphics[width=0.3\textwidth]{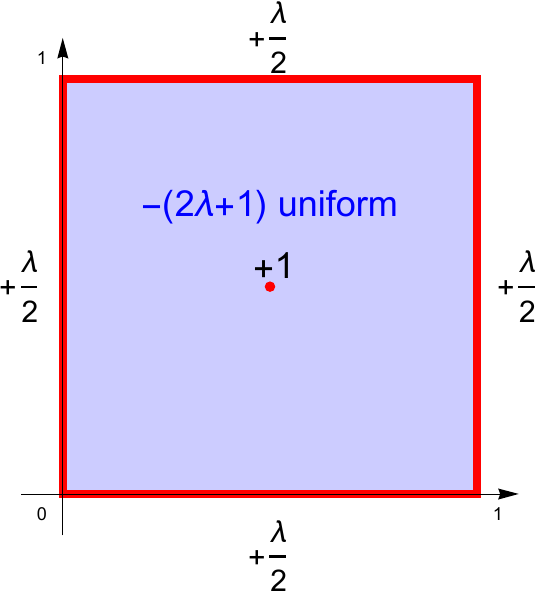}
    \caption{Left: A visualization of the single-dimensional signed measure under adverse selection, for a uniform valuation distribution, $\lambda=\frac{1}{3}$ and initial valuation $c=\frac{1}{2}$. Right: A visualization of the two-dimensional signed measure under adverse selection, for uniform valuation distribution and initial valuation $c=(\frac{1}{2}, \frac{1}{2})$. There is positive mass $\frac{\lambda}{2}$ distributed along the boundaries and a point mass at the point $c$ and $-(2\lambda + 1)$ negative mass spread uniformly through the space. %
    \label{fig:adversemeasure}}
\end{figure}

\citet{MilionisMyersonianFrameworkOptimal2023} give optimal mechanisms for the 1-dimensional case, that is, for the case of a market marker trading a single kind of good for a numeraire.
 They show that optimal mechanisms in this case take the simple form of a bid-ask spread:
 an offer to sell one unit of a good at some price, and buy at a higher price. By analogy to a Myerson auction~\cite{Myerson1981Optimal}, they set these prices to maximize a ``virtual valuation'' which depends on the type distribution.

Their problem setting is the one-dimensional case of the framework we study here, so we can view their solutions as convex utility functions lying in $\mathcal{U}$. As a test and demonstration of our framework, we can construct dual solutions whose transport cost matches the expected profit, establishing optimality in an alternate way.

\subsection{Constructing the transformed measure}
Consider the case where the mechanism designer updates their beliefs about prices by linearly interpolating, so that $\pi(c, x) = \lambda c + (1 - \lambda) x$, and distributions are uniform.

We can consider the transformed measure,  which places mass as  follows: 
\begin{itemize}
    \item $-(n+1)=-2$ uniformly.
    \item $+(1-\lambda)$ uniformly from the $f(x) \divv \pi = \pi'$ term.
    \item The above cancels out to $-(\lambda + 1)$ uniformly.
    \item $+\lambda c$ mass at 0 (left normal term).
    \item $+(\lambda - \lambda c)$ mass at 1 (right normal term).
    \item +1 mass at $c$.
\end{itemize}
The total mass of $\mu^+$ and $\mu^-$ are each $\lambda + 1$, so $\mu$ integrates to 0.
For a visualization of this measure, see Fig. \ref{fig:adversemeasure}.

\subsection{Conjecture of optimal mechanism}  

Following~\citet{MilionisMyersonianFrameworkOptimal2023}, the upper and lower virtual value functions are $\lambda(x - c) - (1 - x)$ and $\lambda(c - x) - x$, so selling and buying at $(1+\lambda c)/(\lambda + 1)$ and $(\lambda / (\lambda + 1)) c$ respectively is optimal.
The expected profit of this mechanism, calculated directly,
 is $\frac{(2 (c-1) c+1) \lambda ^2}{2 (\lambda +1)}$.

\subsection{Constructing the dual certificate} 

$\mu^-$ thus has a CDF of $(1+\lambda)x$, i.e., it is uniform.
 Again, to construct $\gamma_1$, we can perform preprocessing 
 on $\mu^+$ to spread the point mass at $c$ out to a uniform mass on the interval $[\frac{c \lambda}{1 + \lambda}, \frac{1+ c \lambda}{1+\lambda}]$, resulting in CDF
\begin{equation*}
    C_{\gamma_1} = \begin{cases}
 0 & x<0 \\
 c \lambda  & 0\leq x<\frac{c \lambda }{\lambda +1} \\
 c \lambda +\frac{x-\frac{c \lambda }{\lambda +1}}{\frac{c \lambda +1}{\lambda
   +1}-\frac{c \lambda }{\lambda +1}} & \frac{c \lambda }{\lambda +1}\leq x<\frac{c
   \lambda +1}{\lambda +1} \\
 c \lambda +1 & \frac{c \lambda +1}{\lambda +1}\leq x<1 \\
 \lambda +1 & 1\leq x
\end{cases}.
\end{equation*}
Letting $\gamma_2 = \mu^-$, we have a transport cost of $\int_{-\infty}^{\infty} |C_{\gamma_1} - C_{\gamma_2}| = \frac{(2 (c-1) c+1) \lambda ^2}{2 (\lambda +1)}$, which matches the expected profit, proving the proposed mechanism optimal.

Observe that in the above case, the point mass $\mathbb{I}_A(c)$ is spread out via preprocessing at no transport cost. The endpoints of the area to which it is spread precisely correspond to the no-trade region which is inside the bid and ask prices.

\section{Optimal mechanisms in multi-parameter settings}
\label{sec:optimalmulti}

Having warmed up on a 1D case where the optimal solution was already known, 
we proceed to apply the dual transport
framework in earnest.
This is not so easy to do. We have established  that for any optimal mechanism, its expected profit must be matched by the transport cost of the dual transport problem.
However, this gives us little insight into what an optimal mechanism might look like for a given problem instance; 
in particular, any combination of allocations and prices could be optimal.

This is where we turn to differentiable economics to search for profitable mechanisms.
Without any prior knowledge about the structure of the optimal mechanism, but only sample access to the valuation distribution, we show that we can use these techniques from differentiable economics to
 find high-performing mechanisms.
We treat the results coming from this computational pipeline as conjectures for the form of the optimal mechanism. Although they are necessarily only numerical approximations, they contain a lot of useful information that we can use, along with knowledge about the dual program, to come up with plausible analytic expressions for the menus and prices. Then, we construct a transport plan for the dual and show that the primal and dual values match.

This general approach has been used before to find a few new optimal mechanisms in 
simple, single bidder auction problems~\cite{Shen2019Automated,duetting2023optimalJACM}. Here, though, we apply differentiable economics as a tool throughout our work, adopting it to develop intuitions and theoretical directions  for a new problem from the very start of our work (in fact, simple experiments with differentiable economics were able to yield helpful clues before we began proving mechanisms to be optimal, in particular providing hints about the geometric interpretation of the relation $\succeq$).

\subsection{Differentiable economics: Architecture and training methods}

We encode our convex utility function (in the style of \textit{RochetNet} from \citet{duetting2023optimalJACM}) as a maximization over a large number of hyperplanes, each corresponding to an allocation $a_i$ (represented by a parameter $\alpha \in \mathbb{R}$ which is scaled via a sigmoid transformation to encode a feasible allocation in $[-1,1]$)and payment $p_i$:
\begin{equation*}
    u(x) = \max_i a_i \cdot x - p_i; \quad \mbox{where}\ a_i = 2 \sigma(\alpha) - 1, \alpha_i \in \mathbb{R}, p_i \in \mathbb{R}.
\end{equation*}

Motivated by empirical results \cite{duetting2023optimalJACM,curry2023differentiable} as well as theory \cite{hertrich2023mode}, we overparameterize our networks with $2^{10}$ menu items, even though 
 few menu items are typically used. We train to directly optimize the profit objective, using softmax with a temperature of 100 as a differentiable surrogate for the max operation, and use Adam with the typical learning rate of $10^{-3}$ and large batch sizes of $2^{16}$ or $2^{15}$. Each job was scheduled on a compute cluster, running with a single A100 GPU.

\subsection{A new family of optimal mechanisms}
\label{ssec:uniform2d}

We consider the  case of two 
kinds of goods, with values uniformly distributed on $[0,1]$, a market maker with initial valuation of $c=(\frac{1}{2}, \frac{1}{2})$, and a linear update model, $\pi(c, x) = \lambda c + (1 - \lambda)x$. We present a new family of optimal mechanisms for this setting. %

We train models for a range of values of $\lambda$,
 some of which are shown in \Cref{fig:noiserochet}.
These mechanisms have some shared structure. There is symmetry when reflecting across $y=x$, $y=\frac{1}{2}$, and $x=\frac{1}{2}$. 
 They only make deterministic allocations (i.e., either selling -1, 0, or 1 units of each good), and within deterministic allocations, they offer every possible allocation for some price.
There is a no-trade region centered at the initial price $c=(\frac{1}{2}, \frac{1}{2})$, and 
as we would expect, this increases in size (i.e., the prices for the trader get worse) as
 adverse selection becomes stronger.

\subsubsection{Constructing the transformed measure}

The transformed measure (Eq.~\ref{eq:mechobjectiverepeat}) for this case is:
\begin{itemize}
    \item $-(2\lambda + 1)$ mass distributed uniformly.
    \item A $+1$ point mass at $c$.
    \item Mass of $+(1-c_1)\lambda$ on the right boundary, $+(1-c_2)\lambda$ on the top boundary, $+ c_1\lambda$ on the left boundary, and $+c_2\lambda$ on the bottom boundary.
\end{itemize}
The positive and negative masses each have magnitude $2\lambda + 1$. The measure for the case where $c=(\frac{1}{2}, \frac{1}{2})$ is visualized in Fig. \ref{fig:adversemeasure}.

\begin{figure}
    \centering
    \includegraphics[width=0.9\textwidth]{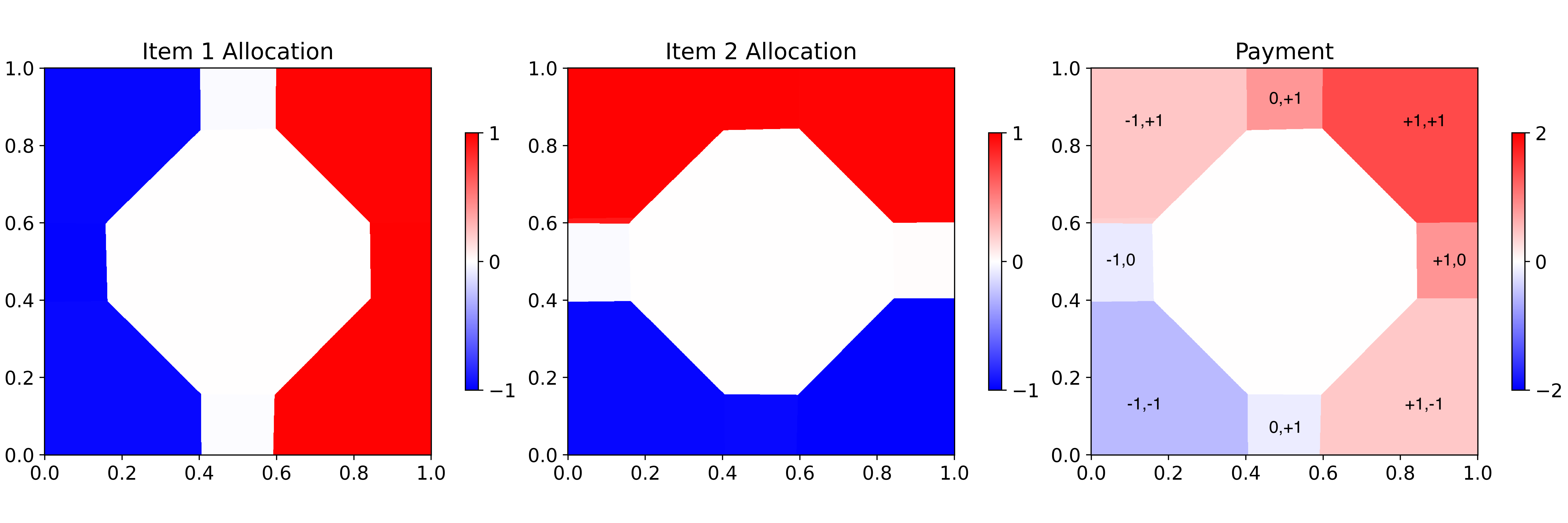}
    \includegraphics[width=0.9\textwidth]{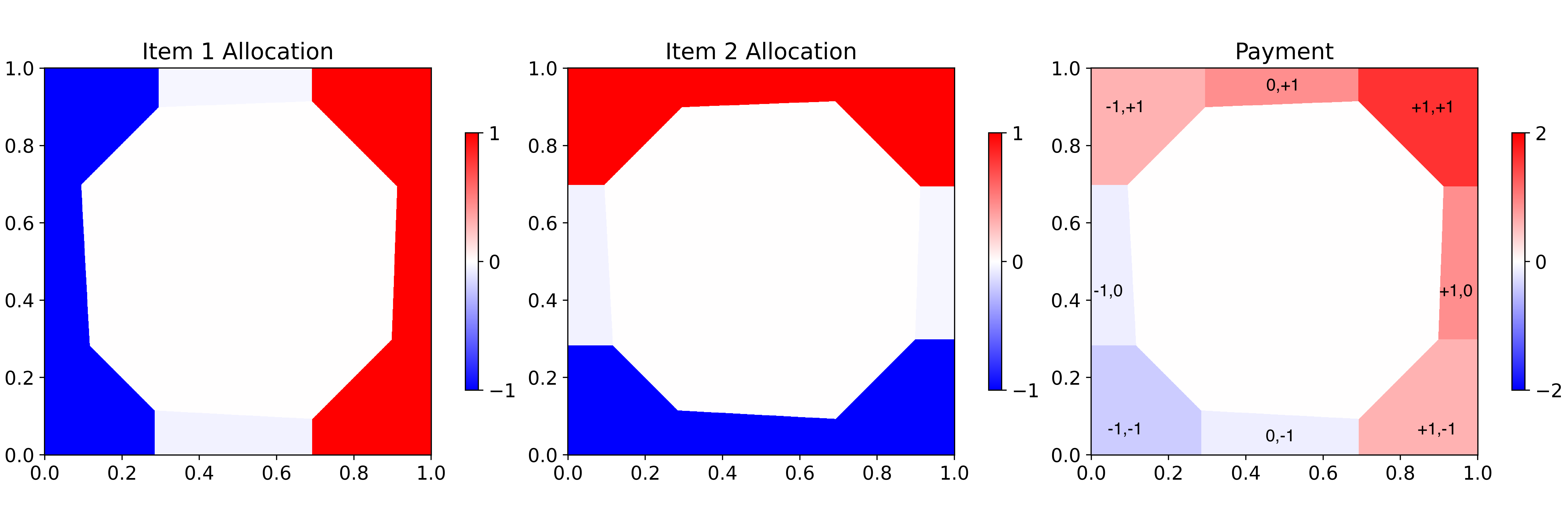}
    \includegraphics[width=0.9\textwidth]{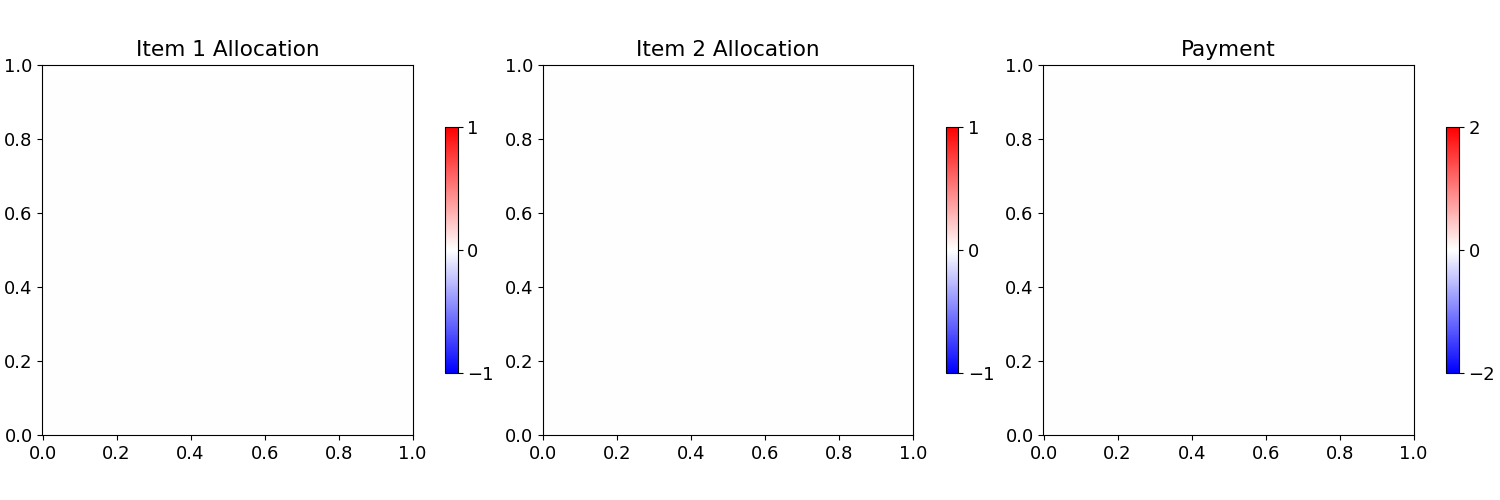}
    \caption{Learned allocation and payment rules for two kinds of goods
    and $c=(\frac{1}{2}, \frac{1}{2})$, where trader valuations are distributed uniformly. Top to bottom, $\lambda = 0$ (no adverse selection), $\lambda=\frac{1}{2}$, and $\lambda=1$ (full adverse selection, so no trade is desirable, hence the blank plot). The axes of the plot are the trader's valuation for either good. Each distinct region is associated with a specific menu item; these menu items are marked on the payment rule plots. %
    \label{fig:noiserochet}}
\end{figure}

\subsubsection{Conjecture for optimal mechanism}

Consider the case where $c=(\frac{1}{2}, \frac{1}{2})$. We first train a model to maximize profit; the allocation rules for the mechanism after training are shown in the top row of Fig. \ref{fig:noiserochet}.
We use the learned model, as well as our knowledge of the problem setting, to conjecture the structure of the optimal mechanism. 

\begin{quote}
{\em The mechanism should be deterministic, and offer eight menu items in addition to the no-trade option (every combination of buying or selling a full unit of the items, or neither). There should be a large octagonal no-trade region in the center. Also, because valuations are symmetric, the mechanism should be symmetric, despite minor asymmetries in the learned solution which can be ascribed to numerical imprecision.}
\end{quote}

\subsubsection{Constructing a dual certificate}

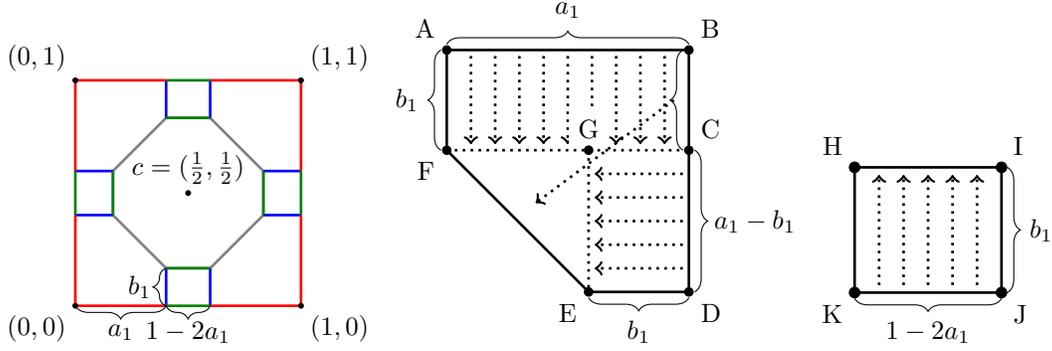
\begin{figure}
\centering
\begin{tikzpicture}[scale=3]

\coordinate (V) at (0.5, 0.5);

\coordinate (A) at (0,1);
\coordinate (B) at (1,1);
\coordinate (C) at (1,0);
\coordinate (D) at (0,0);

\coordinate (E) at ({\ra}, {1 - \rb});
\coordinate (F) at ({1 - \ra}, {1 - \rb});
\coordinate (G) at ({1 - \rb}, {1 - \ra});
\coordinate (H) at ({1 - \rb}, {\ra});
\coordinate (I) at ({1 - \ra}, {\rb});
\coordinate (J) at ({\ra}, {\rb});
\coordinate (K) at ({\rb}, {\ra});
\coordinate (L) at ({\rb}, {1 - \ra});

\coordinate (M) at ({\ra}, {1});
\coordinate (N) at ({1 - \ra}, {1});
\coordinate (O) at ({1}, {1 - \ra});
\coordinate (P) at ({1}, {\ra});
\coordinate (Q) at ({1 - \ra}, {0});
\coordinate (R) at ({\ra}, {0});
\coordinate (S) at ({0}, {\ra});
\coordinate (T) at ({0}, {1 - \ra});

\draw (A) -- (B) -- (C) -- (D) -- cycle;
\draw [line width=1pt, gray] (E) -- (F) -- (G) -- (H) -- (I) -- (J) -- (K) -- (L) -- cycle;
\draw [line width=1pt, blue]
    (M) -- (E)
    (N) -- (F)
    (O) -- (G)
    (P) -- (H)
    (Q) -- (I)
    (R) -- (J)
    (S) -- (K)
    (T) -- (L);
\draw [line width=1pt, red] (A) -- (M)
    (N) -- (B)
    (B) -- (O)
    (P) -- (C)
    (C) -- (Q)
    (R) -- (D)
    (D) -- (S)
    (T) -- (A);
\draw [line width=1pt, ao(english)]
    (M) -- (N)
    (O) -- (P)
    (Q) -- (R)
    (S) -- (T)
    (E) -- (F)
    (G) -- (H)
    (I) -- (J)
    (K) -- (L);

\node [above left = 0pt] at (A) {$(0, 1)$};
\node [above right = 0pt] at (B) {$(1, 1)$};
\node [below right = 0pt] at (C) {$(1, 0)$};
\node [below left = 0pt] at (D) {$(0, 0)$};
\node [above = 0pt] at (V) {$c=(\frac{1}{2}, \frac{1}{2})$};
\filldraw (A) circle[radius=0.3pt];
\filldraw (B) circle[radius=0.3pt];
\filldraw (C) circle[radius=0.3pt];
\filldraw (D) circle[radius=0.3pt];
\filldraw (V) circle[radius=0.3pt];

\draw [decorate,decoration={brace,amplitude=5pt,mirror,raise=0ex}] (D) -- (R) node[midway,yshift=-1em]{$a_1$};
\draw [decorate,decoration={brace,amplitude=5pt,raise=0ex}] (R) -- (J) node[midway,xshift=-1em]{$b_1$};
\draw [decorate,decoration={brace,amplitude=5pt,mirror,raise=0ex}] (R) -- (Q) node[midway,yshift=-1em]{$1 - 2a_1$};

\end{tikzpicture}
\begin{tikzpicture}[scale=8]

\coordinate (A) at ({0}, {\ra});
\coordinate (B) at ({\ra}, {\ra});
\coordinate (C) at ({\ra}, {\ra - \rb});
\coordinate (D) at ({\ra}, {0});
\coordinate (E) at ({\ra - \rb}, {0});
\coordinate (F) at ({0}, {\ra - \rb});
\coordinate (G) at ({\ra - \rb}, {\ra - \rb});

\draw [line width = 1pt] (A) -- (B) -- (D) -- (E) -- (F) -- cycle;
\draw [line width = 1pt, dotted] (F) -- (C);
\draw [line width = 1pt, dotted] (G) -- (E);

\foreach \i in {1,...,9}
{
    \draw [->, dotted, line width = 1pt] ({\ra / 10 * \i}, {\ra - 0.01}) -- ({\ra / 10 * \i}, {\ra - \rb + 0.01});
}
\foreach \i in {1,...,5}
{
    \draw [->, dotted, line width = 1pt] ({\ra - 0.01}, {(\ra - \rb) / 6 * \i}) -- ({\ra - \rb + 0.01}, {(\ra - \rb) / 6 * \i});
}

\draw [decorate,decoration={brace,amplitude=5pt,mirror,raise=0.5ex}] (B) -- (C) node[midway,yshift=-1em]{};
\draw [->, dotted, line width = 1pt] ({\ra - 0.025}, {\ra - \rb / 2}) -- ({(\ra - \rb) / 2 + 0.03}, {(\ra - \rb) / 2 + 0.03});

\node [above left = 1pt] at (A) {A};
\filldraw (A) circle[radius=0.2pt];
\node [above right = 1pt] at (B) {B};
\filldraw (B) circle[radius=0.2pt];
\node [above right = 1pt] at (C) {C};
\filldraw (C) circle[radius=0.2pt];
\node [below right = 1pt] at (D) {D};
\filldraw (D) circle[radius=0.2pt];
\node [below left = 1pt] at (E) {E};
\filldraw (E) circle[radius=0.2pt];
\node [below left = 1pt] at (F) {F};
\filldraw (F) circle[radius=0.2pt];
\node [above = 1pt, fill=white] at (G) {G};
\filldraw (G) circle[radius=0.2pt];

\draw [decorate,decoration={brace,amplitude=5pt,raise=0.5ex}] (A) -- (B) node[midway,yshift=1.5em]{$a_1$};
\draw [decorate,decoration={brace,amplitude=5pt,mirror,raise=0.5ex}] (A) -- (F) node[midway,xshift=-1.5em]{$b_1$};
\draw [decorate,decoration={brace,amplitude=5pt,raise=0.5ex}] (C) -- (D) node[midway,xshift=2.5em]{$a_1 - b_1$};
\draw [decorate,decoration={brace,amplitude=5pt,mirror,raise=0.5ex}] (E) -- (D) node[midway,yshift=-1.5em]{$b_1$};

\end{tikzpicture}
\begin{tikzpicture}[scale=10]

\coordinate (H) at ({0}, {\rb});
\coordinate (I) at ({1 - 2 * \ra}, {\rb});
\coordinate (J) at ({1 - 2 * \ra}, {0});
\coordinate (K) at ({0}, {0});

\draw [line width = 1pt] (H) -- (I) -- (J) -- (K) -- cycle;

\foreach \i in {1,...,5}
{
    \draw [->, dotted, line width = 1pt] ({(1 - 2 * \ra) / 6 * \i}, {0 + 0.01}) -- ({(1 - 2 * \ra) / 6 * \i}, {\rb - 0.01});
}

\node [above left = 1pt] at (H) {H};
\filldraw (H) circle[radius=0.2pt];
\node [above right = 1pt] at (I) {I};
\filldraw (I) circle[radius=0.2pt];
\node [below right = 1pt] at (J) {J};
\filldraw (J) circle[radius=0.2pt];
\node [below left = 1pt] at (K) {K};
\filldraw (K) circle[radius=0.2pt];

\draw [decorate,decoration={brace,amplitude=5pt,raise=0.5ex}] (J) -- (K) node[midway,yshift=-1.5em]{$1 - 2 a_1$};
\draw [decorate,decoration={brace,amplitude=5pt,raise=0.5ex}] (I) -- (J) node[midway,xshift=1.5em]{$b_1$};

\end{tikzpicture}
\caption{Left: Partition of the unit square that corresponds to the optimal transport solution to the dual problem in the two-item case, for the specific values $\lambda=1$ (pure noise trading) and $c=(\frac{1}{2}, \frac{1}{2})$. The red edges are of the same length, denoted by $a_1$, while the blue edges are of the same length as well, denoted by $b_1$. The dark green edges share the same length $1 - 2a_1$. Right: A visualization of the optimal transport solution for the partitioned pentagon-shaped and rectangle regions in the two-item case, for the specific values $\lambda=1$ (pure noise trading) and $c=(\frac{1}{2}, \frac{1}{2})$.
\label{fig:noise_trading_partition_and_plan}}
\end{figure}

We now show how to construct a dual certificate and explicitly calculate its transport cost. For clarity, we show here the ``noise trading'' case where $c=(\frac{1}{2}, \frac{1}{2})$ and $\lambda=1$. We also develop
 the optimal mechanism design parameterized by $\lambda$ when $c=(\frac{1}{2}, \frac{1}{2})$. The calculations for this are similar but uglier, and  relegated to~\Cref{app:linearbelief}.

\paragraph{Partitioning the type space} The solution of the optimal transport dual problem can be illustrated using an optimal partition of the unit square so that the minimal total transport cost is the sum of the minimal transport cost for each of the partitioned regions. Fig.~\ref{fig:noise_trading_partition_and_plan}, left, shows the optimal partition in this case. In the optimal transport solution, the $+1$ point mass at $c=(\frac{1}{2}, \frac{1}{2})$, according to our definition of $\succeq$, can be spread outward (away from $c$) to uniformly cover the octagon-shaped region at the center, at zero cost. In addition to the octagon, there are four equivalent rectangles and four equivalent pentagons. Each of the rectangles and pentagons contains a part of the four outer edges of the unit square, and the positive mass on that part of the edge gets distributed uniformly over the rectangle or pentagon to cancel out the negative mass. The total transport cost is therefore
\(
    C_1 = 4 C_1^{(R)} + 4 C_1^{(P)},
\)
where $C_1^{(R)}$ is the transport cost of one rectangle, $C_1^{(P)}$ is the transport cost of one pentagon, and the subscript $1$ represents that these quantities are for the case of $\lambda = 1$.

\paragraph{Transport plan} We parameterize the partition in terms of lengths $a_1$ and $b_1$. In terms of these quantities, the transport cost for each rectangle is 
\(
    C_1^{(R)} = \int_{0}^{b_1} \left(\frac{1-2a_1}{2}\right) \frac{1}{b_1} x \mathrm{d}x = \frac{(1-2a_1)b_1}{4},
\)
while the total transport cost for the pentagon is $C_1^{(P)} = \frac{(5a_1 + b_1) b_1}{6}$
(see the detailed calculations in Appendix~\ref{app:2dnoise}).
The actual transport plans are visualized in Fig. \ref{fig:noise_trading_partition_and_plan}, right and  involve moving the positive mass along the edges uniformly inward until the boundary of the no-trade region is reached.

\paragraph{Solving for $a_1$ and $b_1$}
The structure of the transport plan constrains the values of $a_1$ and $b_1$.
There is a positive mass on edge JK of length $1-2a_1$, and as the mass density per unit length of the edge is $\frac{1}{2}$, the mass is $\frac{1-2a_1}{2}$. This mass is moved upward to be distributed uniformly over the rectangle HIJK and should perfectly cancel out the negative mass in the rectangle. 
The density of the initial negative mass over the rectangle HIJK is $-3 (1-2 a_1) b_1$, and because it should perfectly cancel out the positive mass $\frac{1-2 a_1}{2}$, we have
\begin{align}
    \frac{1 - 2a_1}{2} -3(1 - 2a_1)b_1 = 0.
\label{eq:noisetrading_ab_eq1}
\end{align}
The positive mass on edge BC is $\frac{b_1}{2}$, and the negative mass on triangle EFG is $-\frac{3(a_1-b_1)^2}{2}$. As these two masses should perfectly cancel out each other, we have
\begin{equation}
    \frac{b_1}{2} - \frac{3(a_1-b_1)^2}{2} = 0.
\label{eq:noisetrading_ab_eq2}
\end{equation}
Similarly, the positive mass on edge CD is $\frac{a_1 - b_1}{2}$, and the negative mass on rectangle CDEG is $-3(a_1-b_1)b_1$. As these two masses should perfectly cancel out, we have
\begin{equation}
    \frac{a_1 - b_1}{2} - 3(a_1-b_1)b_1 = 0.
\label{eq:noisetrading_ab_eq3}
\end{equation}
Combining equations Eq.~\ref{eq:noisetrading_ab_eq1}, Eq.~\ref{eq:noisetrading_ab_eq2}, and Eq.~\ref{eq:noisetrading_ab_eq3}, the values of $a_1$ and $b_1$ can be solved: \( a_1 = \frac{1+\sqrt{2}}{6}, b_1 = \frac{1}{6}\).

\paragraph{Computing the menu}
The boundary between two regions is the point where a trader is indifferent between the two menu items. At any trader type on a boundary, we can easily calculate the welfare of the two neighboring allocations, and infer the price that would make the trader indifferent.
The resulting menu is given in Table \ref{tab:optimalnoisemech} and visualized in Fig. \ref{fig:optimalnoisemech}. 
The total transport cost for this pure noise trading case is 
\begin{equation}
    C_1 = 4 C_1^{(R)} + 4 C_1^{(P)} = 4 \cdot \frac{(1-2a_1)b_1}{4} + 4 \cdot \frac{(5a_1 + b_1) b_1}{6} = \frac{1}{27} \left(6+\sqrt{2}\right) \approx 0.274601,
\end{equation}
which exactly matches the expected profit of the optimal menu, so our solutions are optimal.

\begin{figure}
    \centering
    \includegraphics[width=0.9\textwidth]{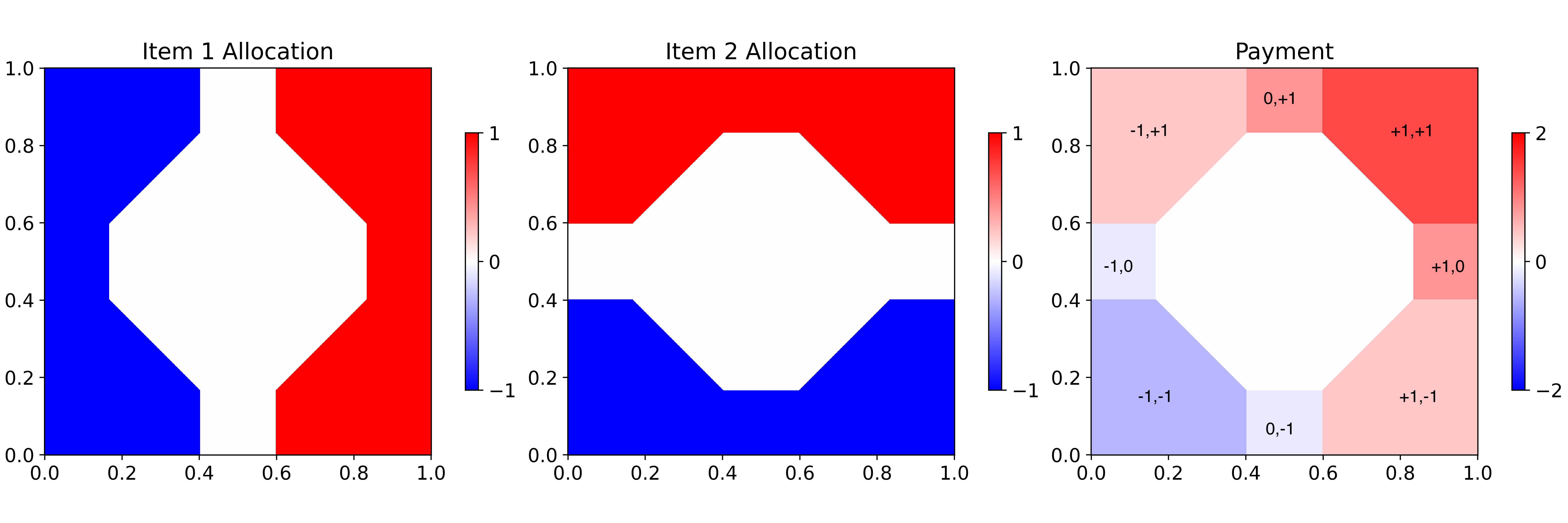}
    \caption{A plot of the %
    allocation and payment rules for the true optimal mechanism under the noise trading ($\lambda=1$) model ($\pi(c, x) = c$) for $c=(\frac{1}{2}, \frac{1}{2})$. Each distinct region is associated with a menu item; these are marked on the payment plot. Compare to the learned mechanism in \Cref{fig:noiserochet}, top.
    \label{fig:optimalnoisemech}}
\end{figure}

\begin{table}[h]
  \centering
  \caption{The optimal menu for the two good, uniform valuation, noise trading (\(\lambda = 1\)) case with \( c=(\frac{1}{2}, \frac{1}{2}). \) }
  \begin{tabular}{@{}l|lllllllll@{}}
    \toprule
    Menu (Item 1) & +1 & 0 & -1 & 0 & +1 & -1 & +1 & -1 & 0\\
    Menu (Item 2) & 0 & +1 & 0 & -1 & +1 & -1 & -1 & +1 & 0\\
    Price & \( \frac{5}{6} \) & \( \frac{5}{6} \) & \( -\frac{1}{6} \) & \( -\frac{1}{6} \) & \( \frac{10-\sqrt{2}}{6} \) & \( \frac{-2-\sqrt{2}}{6} \) & \( \frac{4-\sqrt{2}}{6} \) & \( \frac{4-\sqrt{2}}{6} \) & 0 \\
    \bottomrule
  \end{tabular}
  \label{tab:optimalnoisemech}
\end{table}
\subsubsection{Adverse selection in two dimensions}
As mentioned, the above calculations are for the special case where $\lambda = 1$, i.e. no adverse selection. We also consider the more interesting case of the linear belief updating objective described above, 
for an initial valuation of $c=(\frac{1}{2}, \frac{1}{2})$.
The calculations are similar, though more involved, and are deferred to Appendix~\ref{app:linearbelief}.
The optimal mechanism in this case, as a function of the belief-updating parameter, $\lambda$, is given in Table \ref{tab:noisetradinglambda}.

\begin{table}[h]
\centering
\caption{The optimal menu for the two good,  uniform valuation case with \(c=(\frac{1}{2}, \frac{1}{2})\), for any strength of adverse selection $\lambda$.}
\begin{tabular}{@{}l|lllllllll@{}}
\toprule
Menu (Item 1) & 1 & 0 & -1 & 0 & 1 & -1 & 1 & -1 & 0 \\
Menu (Item 2) & 0 & 1 & 0 & -1 & 1 & -1 & -1 & 1 & 0 \\
Payment & \( \frac{3 \lambda +2}{4 \lambda +2} \) & \( \frac{3 \lambda +2}{4 \lambda +2} \) & \( -\frac{\lambda }{4 \lambda +2} \) & \( -\frac{\lambda }{4 \lambda +2} \) & \( \frac{-\sqrt{2} \lambda +6 \lambda +4}{4 \lambda +2} \) & \( -\frac{\left(2+\sqrt{2}\right) \lambda }{4 \lambda +2} \) & \( \frac{-\sqrt{2} \lambda +2 \lambda +2}{4 \lambda +2} \) & \( \frac{-\sqrt{2} \lambda +2 \lambda +2}{4 \lambda +2} \) & 0 \\
\bottomrule
\end{tabular}
\label{tab:noisetradinglambda}
\end{table}

\subsubsection{Discussion of economic properties of optimal mechanisms}
We now qualitatively discuss the mechanism shown in \Cref{tab:noisetradinglambda}.
There is a no-trade octagon in the center. Even with no adverse selection, this octagon is present, but as adverse selection grows stronger, it grows larger until eventually no trade is possible.
The mechanism engages in mixed bundling: that is, for any combination of buying and selling both goods, there is an offer, as well as distinctly priced offers for buying or selling any individual good. In particular, bundles can involve buying one good and selling another, which we can interpret as taking payment ``in kind.''
The intuition for why this is useful is as follows: considering only the single-good offers in \Cref{tab:noisetradinglambda}, the bid-ask spread is wider than in the optimal mechanism discussed in \Cref{sec:warmup1d}, so when single-good trades take place, they are more profitable, but a larger proportion of traders refuse to participate. However, some of these traders are then won back by offering a discount for the multiple good bundles, so profit is higher on the whole.

We also highlight that the profit gap between this optimal mechanism, and the mechanism consisting of applying the prices from \Cref{sec:warmup1d} separately to each item, is $\frac{\lambda ^3 \left(\left(2 \sqrt{2}-3\right) \lambda +2 \sqrt{2}\right)}{6 (\lambda +1) (2 \lambda +1)^2}$.
For $\lambda=1$ this gives about a 10\% increase in profit; for $\lambda=\frac{\sqrt{2}}{3}$ the increase is highest at 11.4\% -- a fairly substantial improvement in profit.

\subsection{Further Conjectures via Differentiable Economics}

We have presented a family of optimal market maker 
mechanisms for independent, uniformly distributed items, parameterized by the strength of adverse selection.
These mechanisms make use of mixed bundling behavior and we are able to construct dual certificates establishing optimality.

To showcase the flexibility of differentiable economics, we also train models on a much wider range of valuation distributions.
The resulting mechanisms display an interesting range of behavior and 
we conjecture  they are optimal.

\paragraph{Asymmetric noise trading case}
The above examples all considered the 
setting where the initial belief of the market maker is 
centered at $(\frac{1}{2}, \frac{1}{2})$, resulting in a symmetric mechanism.
We also consider the case where the initial belief is centered at $(\frac{1}{3}, \frac{1}{3})$ (again for the case where the type distribution is uniform).
The result of training a mechanism is shown in Figure \ref{fig:onethirdrochet}. Based on this solution, and our knowledge of the dual problem as well as knowledge of the points where bidders must be indifferent between two items, we conjecture (calculations in \Cref{app:offcenterbelief}) an analytic form for this mechanism, whose menu is shown in Table \ref{tab:noisetradingonethird}.%
Here, we observe qualitatively different market maker 
behavior:
 note that the items are bought separately or together (i.e. mixed bundling when the mechanism buys), or bundled with the sale of the other item, but not sold separately (i.e., 
if the trader only wants to purchase, the only choice is the grand bundle).
\begin{figure}
    \centering
    \includegraphics[width=0.9\textwidth]{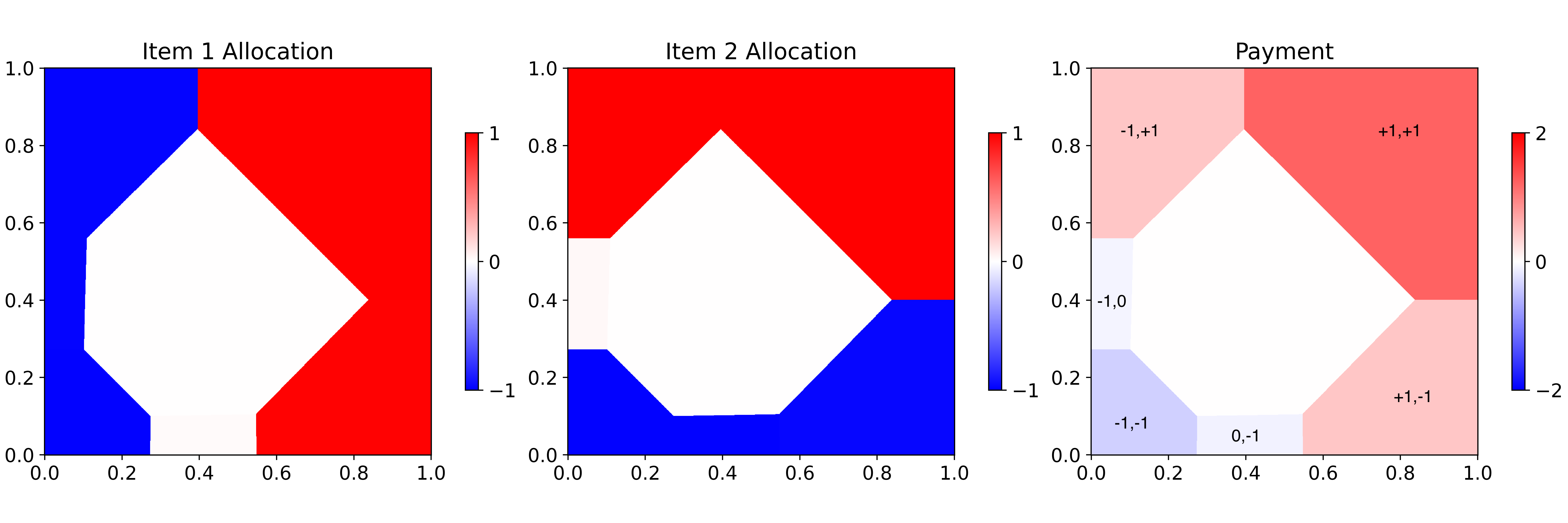}
    \caption{Learned  allocation and payment rules for $c=(\frac{1}{3}, \frac{1}{3})$ for a uniform distribution with no adverse selection ($\lambda=1$). Each distinct region is associated with a menu item; these are marked on the payment plot. A conjectured optimal menu is in \Cref{tab:noisetradingonethird}.
    \label{fig:onethirdrochet}}
\end{figure}

\begin{table}
\caption{A conjectured optimal menu for the adverse selection case for initial valuation $c=(\frac{1}{3}, \frac{1}{3})$, without adverse selection ($\lambda=1$). Unlike the mechanism family in \Cref{tab:noisetradinglambda}, this mechanism never offers to sell either good separately.}
    \begin{tabular}{ccc}
\hline
Menu (Item 1) & Menu (Item 2) & Payment \\
\hline
-1 & 0 & $-\frac{1}{9} \approx -0.111$ \\
0 & -1 & $-\frac{1}{9} \approx -0.1111$ \\
-1 & -1 & $\frac{1}{9}(-2 - \sqrt{2}) \approx -0.3794$\\
1 & -1 & $\frac{2}{63} \left(-\sqrt{2 \left(67-41 \sqrt{2}\right)}-9 \sqrt{2}+\sqrt{2 \left(155+107 \sqrt{2}\right)}+6\right) \approx 0.3143$ \\
-1 & 1 & $\frac{2}{63} \left(-\sqrt{2 \left(67-41 \sqrt{2}\right)}-9 \sqrt{2}+\sqrt{2 \left(155+107 \sqrt{2}\right)}+6\right) \approx 0.3143$ \\
1 & 1 &$\frac{2}{63} \left(\sqrt{2 \left(67-41 \sqrt{2}\right)}-15 \sqrt{2}+\sqrt{2 \left(155+107 \sqrt{2}\right)}+31\right) \approx 1.2413$\\
0 & 0 & $0$ \\
\hline
\end{tabular}
\label{tab:noisetradingonethird}
\end{table}

\paragraph{A three dimensional mechanism}

\begin{figure}
\centering
\includegraphics[width=0.4\textwidth]{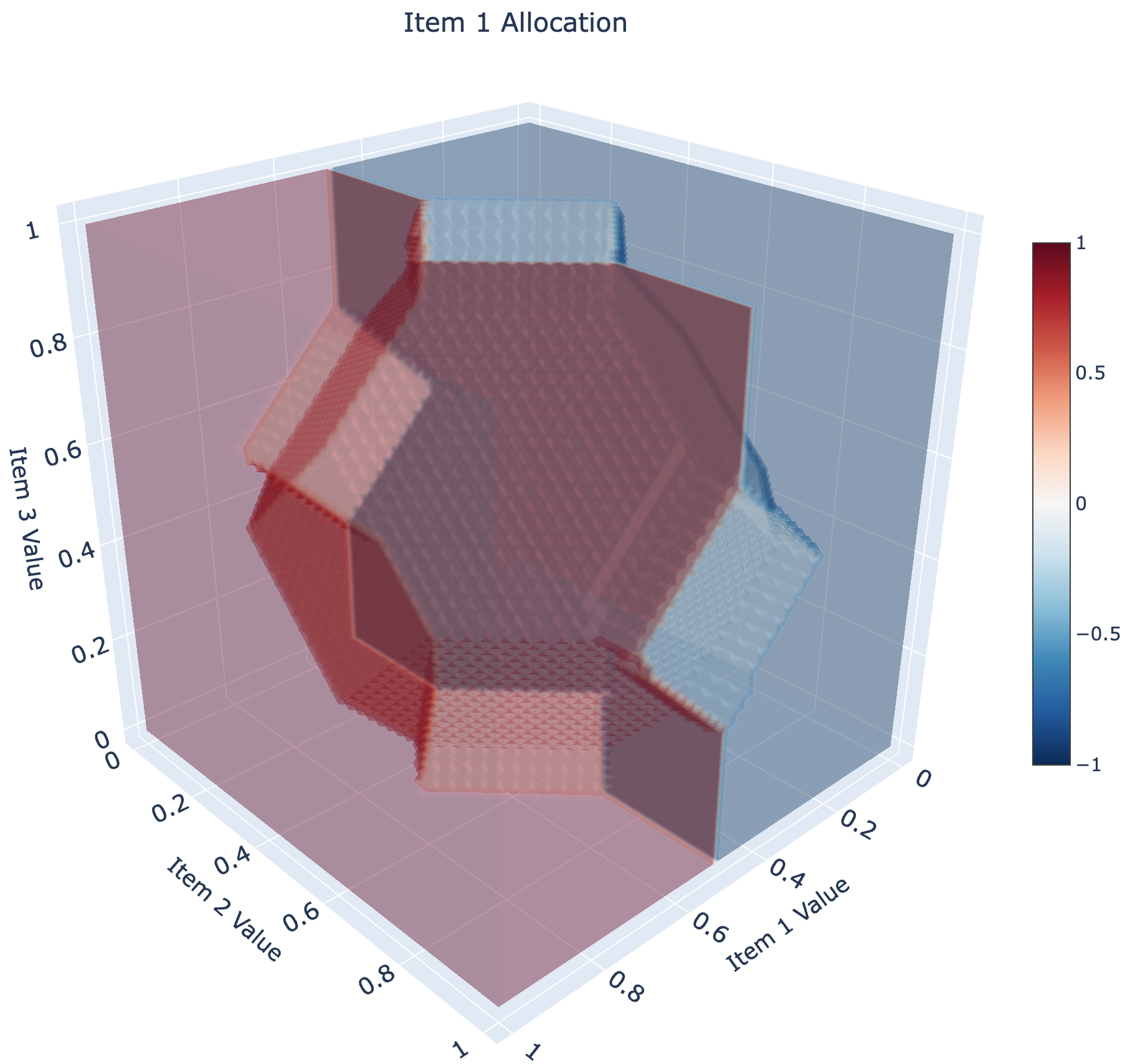}
\includegraphics[width=0.4\textwidth]{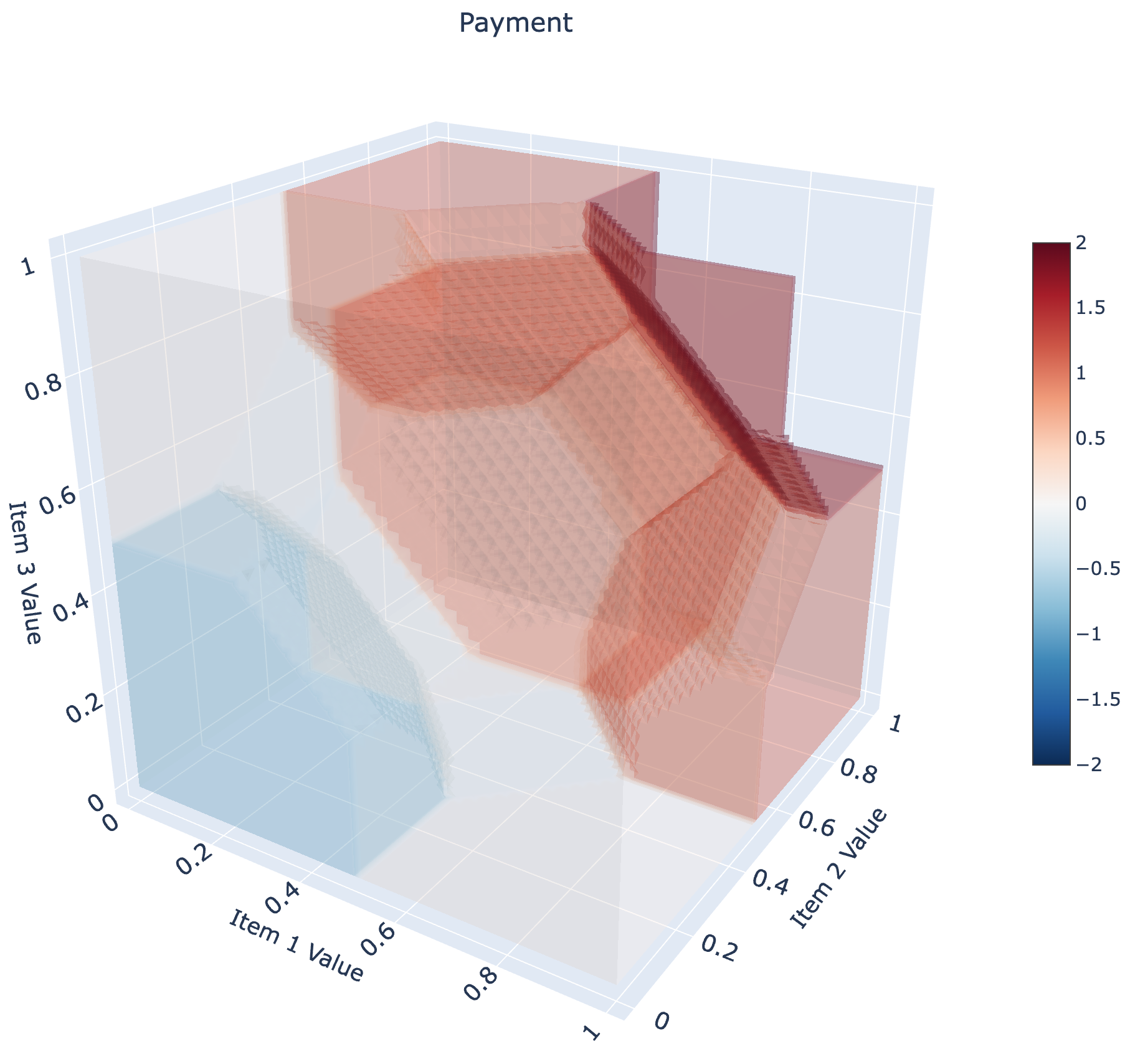}
    \caption{A conjectured optimal mechanism for the uniform distribution with $c=(\frac{1}{2}, \frac{1}{2}, \frac{1}{2})$,
three goods, and no adverse selection.
This visualization shows a density plot of the allocation rule for one
of the goods (the rule for the other goods are similar), and the payment.
    \label{fig:threed}}
\end{figure}
So far we have focused on two-good settings, in which optimal mechanisms already display interesting behavior. Our problem setting is well-defined for any number of goods and we also train a 
{\em three-good mechanism}, as visualized in Figure~\ref{fig:threed}. Two dimensional cross-sectional slices for much of the type space have the same structure as the optimal mechanisms in \Cref{ssec:uniform2d}, although towards the boundaries of the type space the structure changes. The result has extremely interesting and complex structure---a no-trade truncated octahedron (a 14-face polyhedron) in the center, and bundling of all combinations of three goods (buying and selling), along with separate offers for individual goods, but no bundles of two goods.

\paragraph{Beta distribution: infinite menu size}

In the multiple-good monopolist problem, it is well-known that optimal mechanisms under a beta distribution may show infinite menu size~\cite{Daskalakis2017Strong}; in other words, the convex utility function may not be piecewise linear. We train RochetNet under the no-adverse-selection model with $c=(\frac{1}{2}, \frac{1}{2})$, on both Beta(2,1) and Beta(2,2) distributions. Shown in Figure \ref{fig:rochetbeta} and Figure \ref{fig:rochetbetasymmetric}, the resulting allocation rules approximate a continuously varying allocation. For the asymmetric distribution, the mixture only takes place for types that have relatively high valuations in one of the goods.

\begin{figure}
    \centering
    \includegraphics[width=0.9\textwidth]{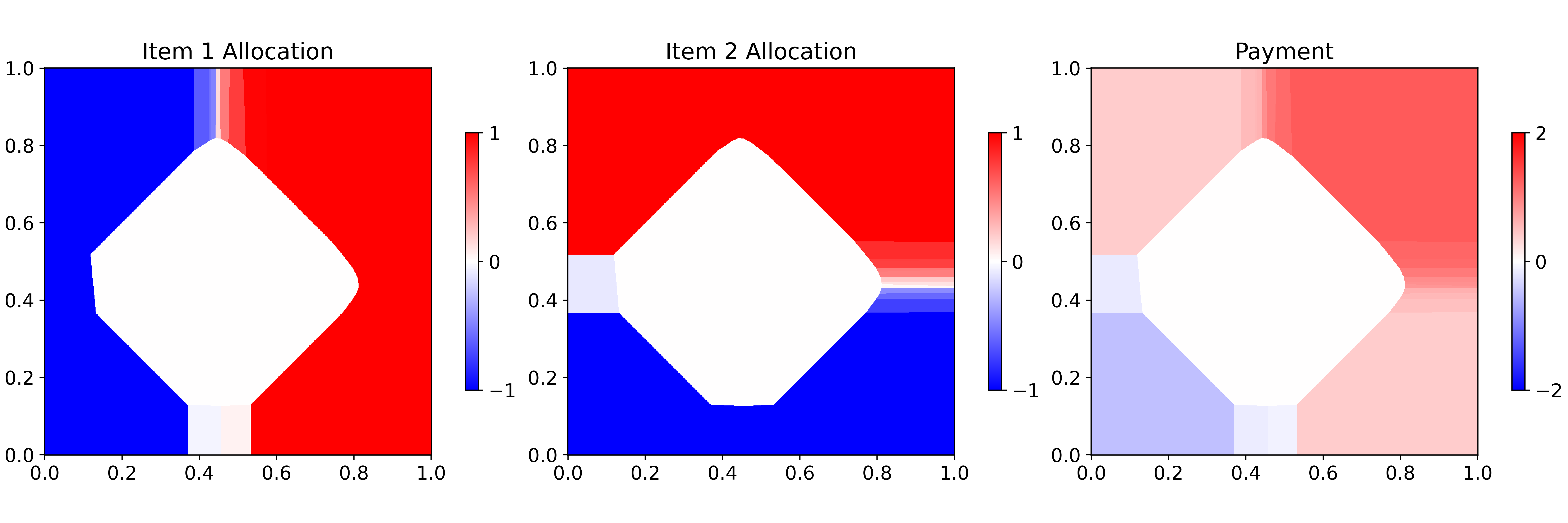}
    \caption{Learned allocation and payment rules for a Beta(1, 2) distribution with $c=(\frac{1}{2}, \frac{1}{2})$ and no adverse selection ($\lambda=1$). There are regions of continuously varying allocation and payment, visible around roughly $x=0.5$ or $y=0.5$, whenever the other valuation is high enough.
    \label{fig:rochetbeta}}
\end{figure}

\begin{figure}
    \centering
    \includegraphics[width=0.9\textwidth]{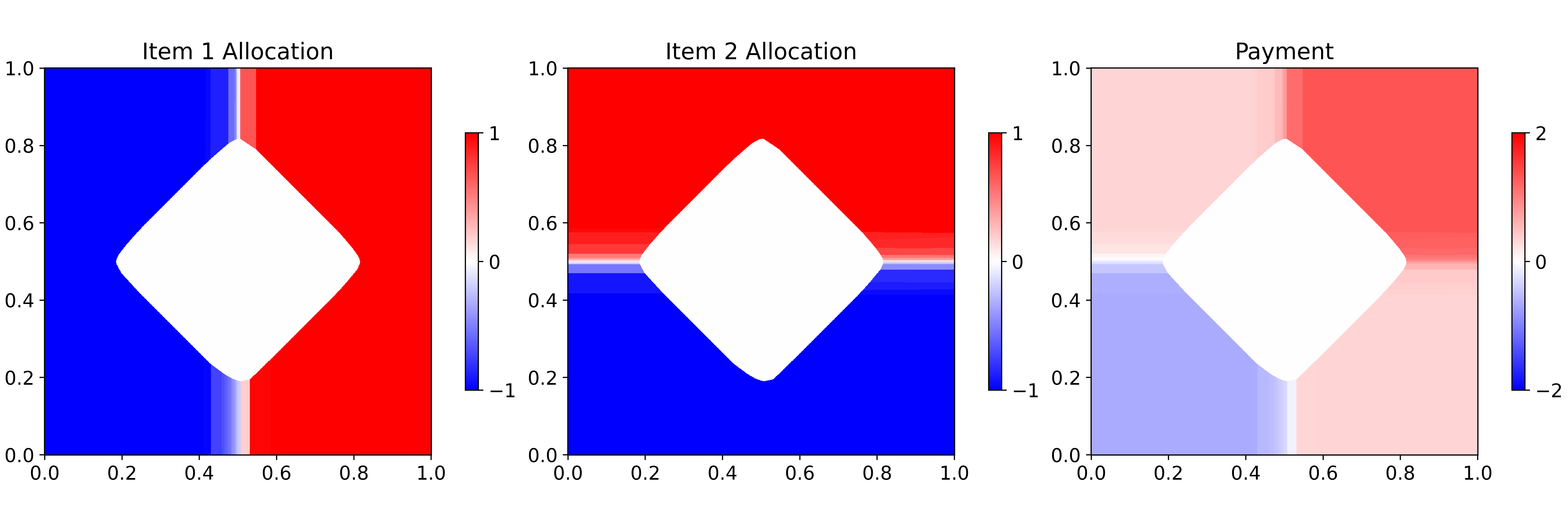}
    \caption{Learned  allocation and payment rules for a symmetric Beta(2, 2) distribution with $c=(\frac{1}{2}, \frac{1}{2})$ and no adverse selection ($\lambda=1$). Here, the mechanism is symmetric, with continuously varying allocations and payments visible around where $x=0.5$ or $y=0.5$ on all sides.
    \label{fig:rochetbetasymmetric}}
\end{figure}

\paragraph{Truncated normal distribution: an interestingly boring mechanism}

We also consider a normal distribution centered and truncated to fit within the type space.
Here, the learned mechanism, shown in \Cref{fig:rochettruncnorm}, is interesting because unlike the other mechanisms we have presented, it is relatively simple.
It is essentially a ``grand bundling'' strategy, making an all-or-nothing offer of the bundle of all goods for some fixed price, albeit with four
 grand bundles, one for each combination of buying and selling.
\begin{figure}
    \centering
    \includegraphics[width=0.9\textwidth]{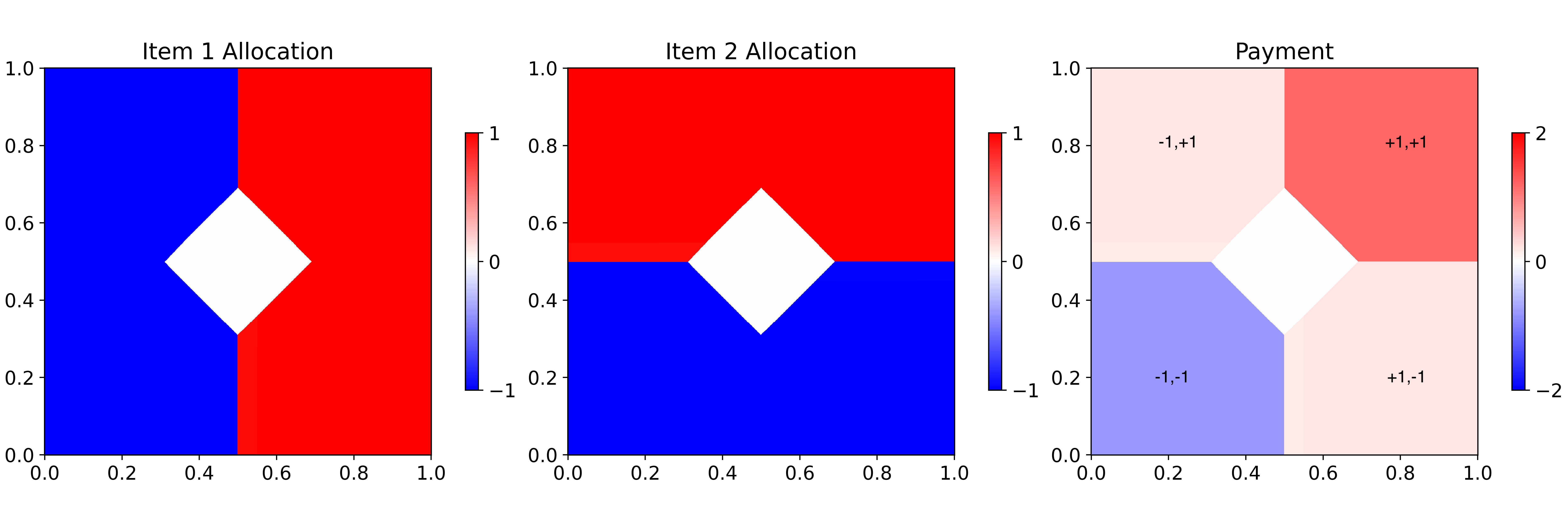}
        \caption{Learned  allocation and payment rules for a truncated normal distribution with mean $\frac{1}{2}$ and standard deviation $\frac{1}{8}$, for $c=(\frac{1}{2}, \frac{1}{2})$. Each distinct region is associated with a menu item; these are marked on the payment plot. This mechanism is quite simple: it does not make any offers of single goods, only trading in bundles.
    \label{fig:rochettruncnorm}}
\end{figure}

\section{Future Generalizations and Extensions}

Here we detail some immediate directions for the generalization of our approach, as well as possible future applications of our techniques.

\paragraph{General constraints on allocations} \citet{Kash2016Optimal} show how to design optimal auctions with fairly arbitrary constraints on the allocation set, such as  requiring deterministic allocations, enforcing minimum quantities of goods to be sold, and many others.
Our computational tools can support such constraints, and an extension of the strong duality results would likely be relatively straightforward.
In particular, if offers must take place in discrete ``lots,'' such an approach could be useful.

\paragraph{Other mechanism design objectives} We focus on a particular adverse selection model. It is relatively general and plausible: one can drop in new update rules $\pi$ with few restrictions~\cite{MilionisMyersonianFrameworkOptimal2023}.
But it is not the only reasonable model of adverse selection, and profit under adverse selection is not the only goal for AMM design.
Our duality results all still hold for any other objective, as long as it can be linearized; i.e.,
converted into an integral against a signed measure $\mu$.
The computational techniques we use also work without modification in this case, so the theoretical and practical machinery here should directly work for many other mechanism design goals.
\paragraph{Trading over time and CFMMs} As discussed in \Cref{app:cfmm}, there is a formal identification between CFMMs and the mechanisms we study, so that any of these mechanisms could be deployed as a CFMM that accepts multiple trades over time. However, after the first trade, the mechanism would no longer be optimal, and re-optimizing after each trade would result in a mechanism that might be vulnerable to roundtrip arbitrage. Optimizing a fixed CFMM under a dynamic model, or optimizing within some space of time-varying CFMMs which still guarantee no arbitrage, is a compelling direction for future work.

\section{Conclusion}

We analyze market-making across multiple goods in the face of adverse selection.
In the settings we study, the market maker has the capacity to buy and sell different bundles at different prices, as well as to accept sales of goods ``in kind'' in lieu of money.
We show that the problem of maximizing profit under adverse selection is dual to an optimal transport problem.
Using our framework, we produce dual certificates for known single-good mechanisms
 and then show how to use differentiable economics to search through the space of mechanisms in settings with multiple goods, in order to generate conjectures for optimal mechanisms.
Based on some of these conjectures, we derive optimal mechanisms and dual certificates.

From an economics perspective, our results establish that in some cases, the \textit{optimal} mechanism 
 for market makers who work across multiple goods must involve making use of their capacity to bundle and to accept goods ``in kind'', resulting in significant improvements in profit.
Further conjectured optimal mechanisms in a range of other cases which show a variety of additional complex behaviors including offers of a continuous spectrum of allocations.
This is by contrast with the single-good case, where a much simpler bid-ask spread is optimal.
It is common in practice to restrict market makers to these simple mechanisms, even when there are multiple goods; our model and results suggest this imposes a cost on them.

From a methodological perspective, we provide theoretical results for another area in the relatively under-explored field of multi-parameter mechanism design, and we show how differentiable economics can be used to guide theory in otherwise very challenging settings. We also hope that these techniques can lead to the design of better market-making mechanisms which will actually be deployed.

\bibliography{references}
\appendix
\section{Connection to constant-function market makers}
\label{app:cfmm}

\citet{Angeris2021Replicating} defines constant-function market makers in terms of a concave \textit{portfolio value function} and a conjugate \textit{trading function}. The portfolio value function $V(x)$ is the value that the market maker expects to retain in their portfolio after an arbitrageur, who knows the correct price vector $x$, engages in a trade. The trading function $\Psi_V(R)$ remains constant for any portfolio that could be the result of a feasible trade; $\Psi_V(R) = -\sup_x (-x^T R - (-V)(x)) = -(-V)^*(-R)$.
\citet{Angeris2021Replicating} require that $V(x)$ is 1-homogeneous; rather than treating payments separately, one good is treated as the numeraire with fixed value normalized to 1. Because $V(x)$ is 1-homogeneous, $\Psi_V(R) = -(-V)^*(-R)$ represents the negative indicator function for a convex set of feasible reserves~\cite{Rockafellar2015Convex}.

Thus there is a close connection between the portfolio value and trading functions, and our formulation of a trader utility $u(x)$ and a menu $u^*(a)$.
For some initial reserves $r_0$, $u(x) = x^T r_0 - V(x)$ (there is some surplus which, after a trade, will be divided between the market maker and trader).
When $u$ is 1-homogeneous and one good is chosen to be the numeraire, $u^*(a)$ will be either 0 or $\infty$ (denoting whether a trade is possible or not).
In our case, we instead allow the output $u^*(a)$ to denote a price for trade $a$, reducing the dimensionality of the valuation vector by 1.

Separately and similarly, \citet{Frongillo2023Axiomatic} shows a connection between CFMMs and scoring rules for prediction markets.

\section{Proofs of duality lemmas}
\label{app:linearize} %

First, we go from an expression for expected profit which involves $\nabla u(x)$ to one that only involves $u$, proving \Cref{lemma:byparts}. The calculation mainly involves repeated application of integration by parts, which we include in extreme detail in this appendix.

\begin{proof}[Proof of \Cref{lemma:byparts}]
First, we can start with the expression for expected profit and slightly rewrite it.
\begin{equation}
\begin{aligned}
    &\int_{\mathcal{X}} \left(\nabla u(x) \cdot (x - \pi(c,x)) - u(x) \right) f(x)\,dx \\
    = & \int_{\mathcal{X}} \left(\sum_{i=1}^n \left(\frac{\partial u}{\partial x_i}\right)  (x_i - \pi_i(c_i, x_i)) - u(x) \right) f(x)\,dx \\
    = & \sum_{i=1}^n \int_{\mathcal{X}} \left(\left(\frac{\partial u}{\partial x_i}\right)  (x_i - \pi_i(c_i, x_i))\right) f(x)\,dx - \int_{\mathcal{X}} u(x) f(x)\,dx.
    \end{aligned}
    \label{eq:beforeintegrationbyparts1}
\end{equation}

Now consider the first term in \ref{eq:beforeintegrationbyparts1}, and pick out the summand for any index $i$. We have 
\begin{equation}
\begin{aligned}
    & \int_{\mathcal{X}} \left(\left(\frac{\partial u}{\partial x_i}\right)  (x_i - \pi_i(c_i, x_i))\right) f(x)\,dx \\
    = & \int_{\mathcal{X}_{-i}} %
    \left( 
    \int_{\mathcal{X}_i} \frac{\partial u}{\partial x_i} x_i f(x)\,dx_{i} - \int_{\mathcal{X}_i} \frac{\partial u}{\partial x_i} \left(\pi_i(c_i,x_i) f(x)\right)\,dx_{i} 
    \right)
    dx_{-i}.
\end{aligned}
    \label{eq:insidesum1}
\end{equation}

Applying integration by parts to these terms, with $X_i = [0, b_i]$ we get
\begin{align*}
    \int_{\mathcal{X}_i} \frac{\partial u}{\partial x_i} x_i f(x)\,dx_{i} &= \left.u(x) x_i f(x)\right|^{x_i = b_i}_0 - \int_0^{b_i} u(x) \left[\frac{\partial}{\partial x_i} \left(x_i f(x)\right)\right]\,dx_i \\
    &= b_i u(b_i, x_{-i})  f(b_i, x_{-i}) - \int_0^{b_i} u(x)\left(x_i \frac{\partial f}{\partial x_i} + f(x) \right)\,dx_i.
\end{align*}
and likewise 
\begin{align*}
    \int_{\mathcal{X}_i} \frac{\partial u}{\partial x_i} \pi_i(c_i,x_i) f(x)\,dx_{i} =& \left.u(x)\pi_i(c_i, x_i) f(x)\right|^{x_i=b_i}_0 - \int_0^{b_i} u(x) \left[ \frac{\partial}{\partial x_i} (\pi_i(c_i, x_i) f(x)) \right]\,dx_i \\
    = &u(b_i, x_{-i}) \pi_i(c_i,b_i)f(b_i, x_{-i}) - u(0, x_{-i}) \pi_i(c_i, 0) f(0, x_{-i}) \\
    &- \int_0^{b_i} u(x) \left(\pi_i(c_i, x_i)\frac{\partial f}{\partial x_i}+f(x) \frac{\partial \pi_i(c_i, x_i)}{\partial x_i}\right)\,dx_i.
\end{align*}

 Putting these back into the expression in \ref{eq:insidesum1}, and distributing through the $\int_{X_{-i}}$, we end up with
\begin{dmath*}
    \int_{\mathcal{X}_{-i}} u(b_i, x_{-i})(b_i - \pi_i(c_i, b_i))f(b_i, x_{-i})\,dx_{-i} - \int_{\mathcal{X}_{-i}} -u(0, x_{-i})\pi_i(c_i, 0)f(0, x_{-i})\,dx_{-i}+ \int_{\mathcal{X}} (\pi_i(c_i,x_i) - x_i) u(x) \frac{\partial f}{\partial x_i}\,dx + \int_{\mathcal{X}} u(x) f(x)\left(\frac{\partial \pi_i(c_i, x_i)}{\partial x_i} - 1\right)\,dx.
\end{dmath*}
for any choice of index $i$.

Putting these terms back inside the sum in \ref{eq:beforeintegrationbyparts1}, we get
\begin{equation*}
    \begin{aligned}
    &\sum_{i=1}^n \int_{\mathcal{X}} \left(\left(\frac{\partial u}{\partial x_i}\right)  (x_i - \pi_i(c_i, x_i))\right) f(x)\,dx - \int_{\mathcal{X}} u(x) f(x)\,dx    \\ 
    =& \sum_{i=1}^n \left(\int_{\mathcal{X}_{-i}} u(b_i, x_{-i})(b_i - \pi_i(c_i, b_i))f(b_i, x_{-i})\,dx_{-i} - \int_{\mathcal{X}_{-i}} -u(0, x_{-i})\pi_i(c_i, 0)f(0, x_{-i})\,dx_{-i}\right) + \\&\sum_{i=1}^n \int_{\mathcal{X}} (\pi_i(c_i,x_i) - x_i) u(x) \frac{\partial f}{\partial x_i}\,dx + \sum_{i=1}^n \int_{\mathcal{X}} u(x) f(x)\left(\frac{\partial \pi_i(c_i, x_i)}{\partial x_i} - 1\right)\,dx - \int_{\mathcal{X}} u(x) f(x)\,dx \\
    =& ... + \sum_{i=1}^n\int_{\mathcal{X}} u(x) f(x)\left(\frac{\partial \pi_i(c_i, x_i)}{\partial x_i} - 1\right)\,dx - \int_{\mathcal{X}} u(x) f(x) \left(1\right)\,dx \\
        =& ... + \sum_{i=1}^n\int_{\mathcal{X}} u(x) f(x)\left(\frac{\partial \pi_i(c_i, x_i)}{\partial x_i}\right)\,dx - (n+1) \int_{\mathcal{X}} u(x) f(x)\,dx\\
    =&\sum_{i=1}^n \left(\int_{\mathcal{X}_{-i}} u(b_i, x_{-i})(b_i - \pi_i(c_i, b_i))f(b_i, x_{-i})\,dx_{-i} - \int_{\mathcal{X}_{-i}} -u(0, x_{-i})\pi_i(c_i, 0)f(0, x_{-i})\,dx_{-i}\right)+ \\ &\int_{\mathcal{X}} u(x)(\nabla f(x) \cdot (\pi(c,x) - x))\,dx +  \sum_{i=1}^n\int_{\mathcal{X}} u(x) f(x)\left(\frac{\partial \pi_i(c_i, x_i)}{\partial x_i}\right)\,dx  - (n+1) \int_{\mathcal{X}} u(x) f(x)\,dx\\
        =& \sum_{i=1}^n \left(\int_{\mathcal{X}_{-i}}  u(b_i, x_{-i})(b_i - \pi_i(c_i, b_i))f(b_i, x_{-i})\,dx_{-i} - \int_{\mathcal{X}_{-i}} -u(0, x_{-i})\pi_i(c_i, 0)f(0, x_{-i})\,dx_{-i}\right)+\\
        &\int_{\mathcal{X}} u(x)(\nabla f(x) \cdot (\pi(c,x) - x))\,dx + \int_{\mathcal{X}} u(x) f(x)\text{div}\pi(c,x)\,dx  - (n+1) \int_{\mathcal{X}} u(x) f(x)\,dx.
    \end{aligned}
\end{equation*}

There are four terms in this sum. The first term can be expressed as a surface integral around $\mathcal{X} = \prod_i [0, b_i]$.  Each of the remaining terms is an integral of $u(x)$ against some density, and we can combine these together. 

After this transformation, we can re-express the above as an integral $\int u\,d\mu'$ against a measure $\mu'$, where
\begin{equation*}
\begin{aligned}
        \mu'(A) &= \int_{\partial X} \mathbb{I}_A(x) f(x)(x - \pi(c, x))\cdot \hat{n}\,dx \\
        &- \int_{\mathcal{X}} \mathbb{I}_A(x) (\nabla f(x) \cdot (x - \pi(c,x)) 
         + (n+1 - \divv \pi(c,x))f(x))\,dx.
        \end{aligned}
\end{equation*}
This relationship holds for any $u$. However, we can plug in $u(x) = 1$ and, after some calculation, find that $\mu'$ integrates to $$\int_{\mathcal{X}} \left(\nabla u(x) \cdot (x - \pi(c,x)) - u(x) \right) f(x)\,dx = \int_{\mathcal{X}} (0 - 1)f(x)\,dx  = \int d\mu' = -1.$$
It will be more convenient to construct a measure $\mu$ that integrates to zero, so that when we split $\mu = \mu^+ - \mu^-$, we face a balanced optimal transport problem between $\mu^+$ and $\mu^-$.
We can take advantage of the fact that for our problem, feasible $u$ have $u(c) = 0$, and let $\mu(A) = \mu'(A) + \mathbb{I}_A(c)$, observing $\int u\,d\mu = \int u\,d\mu'$ for such feasible $u$. Our final measure is thus
\begin{equation*}
\begin{aligned}
        \mu(A) &= \int_{\partial X} \mathbb{I}_A(x) f(x)(x - \pi(c, x))\cdot \hat{n}\,dx \\
        &- \int_{\mathcal{X}} \mathbb{I}_A(x) (\nabla f(x) \cdot (x - \pi(c,x)) 
         + (n+1 - \divv \pi(c,x))f(x))\,dx + \mathbb{I}_A(c).
        \end{aligned}
        \tag{\ref{eq:transformedrepeat}}
\end{equation*}

In conclusion,
\begin{equation*}
    \int_{\mathcal{X}} \left(\nabla u(x) \cdot (x - \pi(c,x)) - u(x) \right) f(x)\,dx = \int u\,d\mu = \int u\,d\mu^+ - \int u\,d\mu^-.
\end{equation*}
for any feasible $u$.
\end{proof}

\begin{proof}[Proof of \cref{lemma:twofunctions}]
First, we rewrite the dual problem in an unconstrained form, proceeding in essentially the same manner as all prior work.
\begin{equation*}
    \inf_{\gamma: \gamma_1 \succeq \mu^+, \gamma_2 \preceq \mu^-} \int \lVert x - y \rVert_1\,d\gamma(x,y) = \inf_\gamma \Theta^*(\gamma) - \Xi^*(\gamma),
\end{equation*}
where
\begin{equation*}
\Theta^*(\gamma) = 
    \begin{cases}
        \int \lVert x - y \rVert_1 \,d\gamma(x, y)  & \text{if } \gamma \in \text{Radon}_+(\mathcal{X} \times \mathcal{X}) \\
        +\infty & \text{, otherwise.}
    \end{cases}
\end{equation*}
and
\begin{equation*}
\Xi^*(\gamma) = \begin{cases}
    0 & \text{if } \gamma_1 \succeq \mu^+, \mu^- \succeq \gamma_2 \\
    -\infty & \text{, otherwise.}
\end{cases}.
\end{equation*}

We will apply Fenchel-Rockafellar duality~\cite{Villani2003Topics}, which states that $\sup_f \Xi(f) - \Theta(f) = \inf_\gamma \Theta^*(\gamma) - \Xi^*(\gamma)$ for concave $\Xi$, and convex and continuous $\Theta$.  The $*$  operation is the (convex or concave) conjugate.
By directly taking the conjugate of $\Theta^*$, we can find
\begin{equation*}
    \Theta^{**}(f) = \begin{cases}
        0 & \text{if } f(x, y) \leq \lVert x - y \rVert_1 \\
        +\infty & \text{, otherwise.}
    \end{cases}
\end{equation*}

This function is proper and convex, and closed because the domain of feasible $f$ is a convex set. Therefore it is legitimate to refer to $\Theta = \Theta^{**}$.
We  propose $\Xi^{**} = \Xi$ as the closed and concave function,
\begin{equation*}
    \Xi^{**}(f) = \begin{cases}
    \int \phi d\,\mu^+ - \int \psi \,d\mu^- & \exists \phi, \psi \in \mathcal{U^\circ}: \phi(x) - \psi(y) = f(x,y) \\
    -\infty & \text{, otherwise,} 
    \end{cases}
\end{equation*}
where $\mathcal{U}^\circ$ is the cone generated by positive scaling of elements of $\mathcal{U}$. This function is proper and concave, and is closed  because the domain of feasible $f$, $\{f \mid \exists \phi, \psi \in \mathcal{U}^\circ : f(x, y) = \phi(x) - \psi(y)\}$ is a closed convex set.
Taking the conjugate again, $\Xi^{***}$ is in fact $\Xi^*$:
\begin{align*}
    \Xi^{***}(\gamma) =& \inf_f \int f(x,y)\,d\gamma(x,y) - \Xi^{**}(f) \\
    =& \inf_{\phi, \psi \in \mathcal{U^\circ}} \int \phi(x)\,d\gamma_1(x) -\int \psi(y)\,d\gamma(y) - \int \phi(x)\,d\mu^+ + \int \psi(y)\,d\mu^- \\
    =& \begin{cases} 0 & \gamma_1 \succeq \mu^+, \gamma_2 \succeq \mu^- \\ -\infty & \text{else} \end{cases} \\ %
    =& \Xi^*(\gamma).
\end{align*}

(If $\gamma_1 \succeq \mu^+$ does not hold, this implies there must exist some $\phi$ for which the integrals are negative; and this can be scaled arbitrarily to $-\infty$. Likewise for $\gamma_2, \mu^-$ and $\psi$. This allows the transformation from an unconstrained to constrained problem.)

For the same reasons as in \citet{Kash2016Optimal}, $\Theta$ is continuous at $f(x,y) = -1$, $\Theta(f) = 0$, and $\Xi(f) = - \int\,d\mu^-$, so the preconditions for Fenchel-Rockafellar duality are satisfied.
At this point, we now have a dual of
\begin{equation*}
    \sup_{\phi(x) - \psi(y) \leq \lVert x - y \rVert_1, \phi, \psi \in \mathcal{U}^\circ} \int \phi\,d\mu^+ - \int \psi\,d\mu^- = \sup_f \Xi(f) - \Theta(f) = \inf_\gamma \Theta^*(\gamma) - \Xi^*(\gamma). 
\end{equation*}
\end{proof}

\section{Noise trading (no adverse selection) in 2 dimensions}
\label{app:2dnoise}

\definecolor{ao(english)}{rgb}{0.0, 0.5, 0.0}
\def\ra{((1 + sqrt(2)) / 6)}
\def\rb{(1 / 6)}

We consider in \S\ref{ssec:uniform2d} a pure noise trading case where the mechanism designer has a fixed belief $c = (\frac{1}{2}, \frac{1}{2})$, which corresponds to setting $\lambda=1$ in the linear belief update model. Here we detail the calculations of the transport cost of the optimal transport plan for the dual problem.

\begin{figure}
    \centering
    \includegraphics[width=0.4\textwidth]{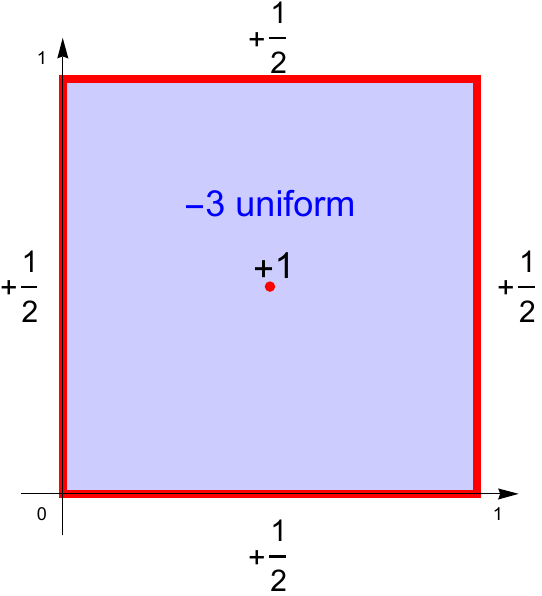}
    \caption{A visualization of the signed measure where $c=(\frac{1}{2}, \frac{1}{2})$, specifically for the noise-trading case where $\lambda=1$.}
    \label{fig:noisetradingsignedmeasure}
\end{figure}
\Cref{fig:noisetradingsignedmeasure} visualizes the signed measure for the optimal transport dual problem for this specific case. There is a point mass of $1$ at $(\frac{1}{2}, \frac{1}{2})$, a uniformly distributed mass of $\frac{1}{2}$ at each of the four edges of the unit square, and a $-3$ mass uniformly distributed over the square. The solution to the dual problem is the minimum transport cost to move the positive mass so that the positive and negative mass cancel each other and the net mass is zero everywhere.

Note that the transport cost here is defined with respect to the Manhattan distance of any mass movement, i.e., measuring the sum of horizontal movement and vertical movement when calculating the transport cost of any movement. As moving positive mass and moving negative mass are symmetric, we can consider the optimal solution in the solution space where only the positive masses are moved to cancel out the uniform negative mass over the unit square.

Based on the learned solution in Fig.~\ref{fig:noiserochet} (top row), we can conjecture the structure of the optimal mechanism and thus the structure of the optimal transport plan, which involves partitioning the type space into separate regions including one octagonal region with zero transport cost, four equivalent rectangular regions, and four equivalent pentagonal regions. 

Fig.~\ref{fig:noise_trading_partition_and_plan}, right, shows the optimal transport for one of the four equivalent rectangles. There is a positive on edge JK, and it is moved upward to uniformly cover the rectangle HIJK. The transport cost of one rectangle is therefore \(C_1^{(R)} = \int_{0}^{b_1} \left(\frac{1-2a_1}{2}\right) \frac{1}{b_1} x \mathrm{d}x = \frac{(1-2a_1)b_1}{4}\).

Fig.~\ref{fig:noise_trading_partition_and_plan}, middle, shows the optimal transport for one of the four equivalent pentagons. There is a positive mass on edges AB and BD. The positive mass at edge AB is moved downward to uniformly cover rectangle ABCF, the positive mass at edge CD is moved leftward to uniformly cover rectangle CDEG, and the positive mass at edge BC is moved to cover uniform triangle EFG. The transport cost for rectangle ABCF is
\begin{equation*}
    \int_{0}^{b_1} \left(\frac{a_1}{2}\right) \frac{1}{b_1} x \mathrm{d}x = \frac{a_1 b_1}{4}.
\end{equation*}
The transport cost for rectangle CDEG is
\begin{equation*}
    \int_{0}^{b_1} \left(\frac{a_1 - b_1}{2}\right) \frac{1}{b_1} x \mathrm{d}x = \frac{(a_1 - b_1) b_1}{4}.
\end{equation*}
For the triangle EFG, the positive mass at edge BC is first moved to point C, then moved to point G, and finally moved to uniformly cover triangle EFG. This is the optimal way to transport the mass as the cost is with respect to the Manhattan distance of movements. Adding the costs of these three transport steps together, the transport for triangle EFG is
\begin{equation*}
    \left(\int_{0}^{b_1} \frac{1}{2} x \mathrm{d}x \right) + \left(\frac{b_1}{2} b_1\right) + \left(\int_{0}^{a_1 - b_1} \int_{0}^{a_1-b_1-x} \left(\frac{b_1}{2}\right) \left(\frac{2}{(a_1 - b_1)^2}\right)  (x+y) \mathrm{d}y \mathrm{d}x\right) = \frac{3b_1^2}{4} + \frac{(a_1 - b_1) b_1}{3}.
\end{equation*}
The total transport cost of the pentagon is 
\begin{equation*}
    C_1^{(P)} = \frac{a_1 b_1}{4} + \frac{(a_1 - b_1) b_1}{4} + \frac{3b_1^2}{4} + \frac{(a_1 - b_1) b_1}{3} = \frac{(5a_1 + b_1) b_1}{6}.
\end{equation*}

\section{Linear belief update in 2 dimensions}
\label{app:linearbelief}

Here, we show the full derivation for a more general case where the mechanism designer has an initial belief of $c$ and updates their belief by linearly interpolating after one trade, so that $\pi(c, x) = \lambda c + (1 - \lambda) x$, $\lambda \in [0, 1]$.
We give a family of transport plans parameterized by $\lambda$, which match a family of mechanisms parameterized by $\lambda$.

\subsection{Constructing the transformed measure}
The transformed measure (Eq.~\ref{eq:mechobjectiverepeat}) for this case is:
\begin{itemize}
    \item $-(2\lambda + 1)$ mass distributed uniformly.
    \item A $+1$ point mass at $c$.
    \item Mass of $+(1-c_1)\lambda$ on the right boundary, $+(1-c_2)\lambda$ on the top boundary, $+ c_1\lambda$ on the left boundary, and $+c_2\lambda$ on the bottom boundary.
\end{itemize}
The positive and negative masses each have magnitude $2\lambda + 1$. The measure for the case where $c=(\frac{1}{2}, \frac{1}{2})$ is visualized in Fig. \ref{fig:adversemeasure} (right).

\subsection{Conjecture for optimal mechanism}
Consider the case where $c=(\frac{1}{2}, \frac{1}{2})$ . We first train a model to maximize profit. Fig.~\ref{fig:noiserochet} visualizes the learned mechanisms when varying the value of $\lambda$.

Similar to the noise trading case, we use the learned solution, as well as our knowledge of the problem setting, to conjecture the structure of the optimal mechanism. The optimal mechanism should be deterministic, offer 8 menu items in addition to the no-trade option, and should be symmetric.

\subsection{Constructing a dual certificate}

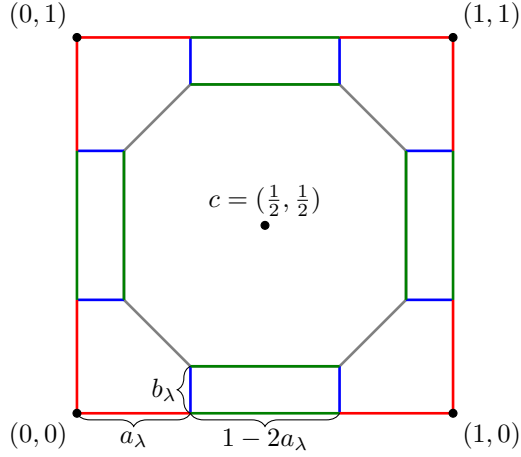
\begin{figure}
\centering
\begin{tikzpicture}[scale=5]

\coordinate (V) at (0.5, 0.5);

\coordinate (A) at (0,1);
\coordinate (B) at (1,1);
\coordinate (C) at (1,0);
\coordinate (D) at (0,0);

\coordinate (E) at ({\ral}, {1 - \rbl});
\coordinate (F) at ({1 - \ral}, {1 - \rbl});
\coordinate (G) at ({1 - \rbl}, {1 - \ral});
\coordinate (H) at ({1 - \rbl}, {\ral});
\coordinate (I) at ({1 - \ral}, {\rbl});
\coordinate (J) at ({\ral}, {\rbl});
\coordinate (K) at ({\rbl}, {\ral});
\coordinate (L) at ({\rbl}, {1 - \ral});

\coordinate (M) at ({\ral}, {1});
\coordinate (N) at ({1 - \ral}, {1});
\coordinate (O) at ({1}, {1 - \ral});
\coordinate (P) at ({1}, {\ral});
\coordinate (Q) at ({1 - \ral}, {0});
\coordinate (R) at ({\ral}, {0});
\coordinate (S) at ({0}, {\ral});
\coordinate (T) at ({0}, {1 - \ral});

\draw (A) -- (B) -- (C) -- (D) -- cycle;
\draw [line width=1pt, gray] (E) -- (F) -- (G) -- (H) -- (I) -- (J) -- (K) -- (L) -- cycle;
\draw [line width=1pt, blue]
    (M) -- (E)
    (N) -- (F)
    (O) -- (G)
    (P) -- (H)
    (Q) -- (I)
    (R) -- (J)
    (S) -- (K)
    (T) -- (L);
\draw [line width=1pt, red] (A) -- (M)
    (N) -- (B)
    (B) -- (O)
    (P) -- (C)
    (C) -- (Q)
    (R) -- (D)
    (D) -- (S)
    (T) -- (A);
\draw [line width=1pt, ao(english)]
    (M) -- (N)
    (O) -- (P)
    (Q) -- (R)
    (S) -- (T)
    (E) -- (F)
    (G) -- (H)
    (I) -- (J)
    (K) -- (L);

\node [above left = 0pt] at (A) {$(0, 1)$};
\node [above right = 0pt] at (B) {$(1, 1)$};
\node [below right = 0pt] at (C) {$(1, 0)$};
\node [below left = 0pt] at (D) {$(0, 0)$};
\node [above = 0pt] at (V) {$c=(\frac{1}{2}, \frac{1}{2})$};
\filldraw (A) circle[radius=0.3pt];
\filldraw (B) circle[radius=0.3pt];
\filldraw (C) circle[radius=0.3pt];
\filldraw (D) circle[radius=0.3pt];
\filldraw (V) circle[radius=0.3pt];

\draw [decorate,decoration={brace,amplitude=5pt,mirror,raise=0ex}] (D) -- (R) node[midway,yshift=-1em]{$a_{\lambda}$};
\draw [decorate,decoration={brace,amplitude=5pt,raise=0ex}] (R) -- (J) node[midway,xshift=-1em]{$b_{\lambda}$};
\draw [decorate,decoration={brace,amplitude=5pt,mirror,raise=0ex}] (R) -- (Q) node[midway,yshift=-1em]{$1 - 2a_{\lambda}$};

\end{tikzpicture}
\caption{Partition of the unit square that corresponds to the optimal transport solution to the dual problem in the two-item case where the initial belief of the mechanism designer is $c=(\frac{1}{2}, \frac{1}{2})$ and after one trade the belief is going to be linearly updated as $\pi(c, x) = \lambda c + (1 - \lambda) x$, $\lambda \in [0, 1]$. The red edges are of the same length, denoted by $a_{\lambda}$, while the blue edges are of the same length as well, denoted by $b_{\lambda}$. The dark green edges share the same length $1 - 2a_{\lambda}$. The exact partitions demonstrated in this figure are for the case of $\lambda = 0.3$.}
\label{fig:linear_update_partition}
\end{figure}

\paragraph{Partitioning the type space} As in the noise trading case, the solution to the dual problem is the minimum transport cost to move the positive mass and cancel out the negative mass perfectly and the optimal transport plan involves partitioning the type space into separate regions including an octagonal no-trade region in the center, four equivalent rectangular regions, and four equivalent octagonal regions. The minimal total transport cost is the sum of the minimal transport cost for each of the partitioned regions. 

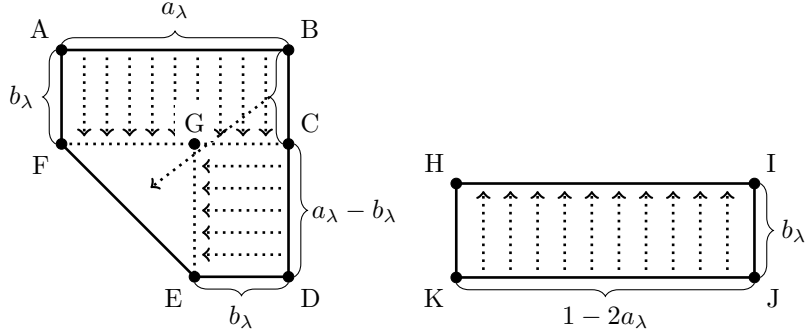
\begin{figure}[t]
\centering
\begin{tikzpicture}[scale=10]

\coordinate (A) at ({0}, {\ral});
\coordinate (B) at ({\ral}, {\ral});
\coordinate (C) at ({\ral}, {\ral - \rbl});
\coordinate (D) at ({\ral}, {0});
\coordinate (E) at ({\ral - \rbl}, {0});
\coordinate (F) at ({0}, {\ral - \rbl});
\coordinate (G) at ({\ral - \rbl}, {\ral - \rbl});

\draw [line width = 1pt] (A) -- (B) -- (D) -- (E) -- (F) -- cycle;
\draw [line width = 1pt, dotted] (F) -- (C);
\draw [line width = 1pt, dotted] (G) -- (E);

\foreach \i in {1,...,9}
{
    \draw [->, dotted, line width = 1pt] ({\ral / 10 * \i}, {\ral - 0.01}) -- ({\ral / 10 * \i}, {\ral - \rbl + 0.01});
}
\foreach \i in {1,...,5}
{
    \draw [->, dotted, line width = 1pt] ({\ral - 0.01}, {(\ral - \rbl) / 6 * \i}) -- ({\ral - \rbl + 0.01}, {(\ral - \rbl) / 6 * \i});
}

\draw [decorate,decoration={brace,amplitude=5pt,mirror,raise=0.5ex}] (B) -- (C) node[midway,yshift=-1em]{};
\draw [->, dotted, line width = 1pt] ({\ral - 0.025}, {\ral - \rbl / 2}) -- ({(\ral - \rbl) / 2 + 0.03}, {(\ral - \rbl) / 2 + 0.03});

\node [above left = 1pt] at (A) {A};
\filldraw (A) circle[radius=0.2pt];
\node [above right = 1pt] at (B) {B};
\filldraw (B) circle[radius=0.2pt];
\node [above right = 1pt] at (C) {C};
\filldraw (C) circle[radius=0.2pt];
\node [below right = 1pt] at (D) {D};
\filldraw (D) circle[radius=0.2pt];
\node [below left = 1pt] at (E) {E};
\filldraw (E) circle[radius=0.2pt];
\node [below left = 1pt] at (F) {F};
\filldraw (F) circle[radius=0.2pt];
\node [above = 1pt, fill=white] at (G) {G};
\filldraw (G) circle[radius=0.2pt];

\draw [decorate,decoration={brace,amplitude=5pt,raise=0.5ex}] (A) -- (B) node[midway,yshift=1.5em]{$a_{\lambda}$};
\draw [decorate,decoration={brace,amplitude=5pt,mirror,raise=0.5ex}] (A) -- (F) node[midway,xshift=-1.5em]{$b_{\lambda}$};
\draw [decorate,decoration={brace,amplitude=5pt,raise=0.5ex}] (C) -- (D) node[midway,xshift=2.5em]{$a_{\lambda} - b_{\lambda}$};
\draw [decorate,decoration={brace,amplitude=5pt,mirror,raise=0.5ex}] (E) -- (D) node[midway,yshift=-1.5em]{$b_{\lambda}$};

\end{tikzpicture}
\begin{tikzpicture}[scale=10]

\coordinate (H) at ({0}, {\rbl});
\coordinate (I) at ({1 - 2 * \ral}, {\rbl});
\coordinate (J) at ({1 - 2 * \ral}, {0});
\coordinate (K) at ({0}, {0});

\draw [line width = 1pt] (H) -- (I) -- (J) -- (K) -- cycle;

\foreach \i in {1,...,10}
{
    \draw [->, dotted, line width = 1pt] ({(1 - 2 * \ral) / 11 * \i}, {0 + 0.01}) -- ({(1 - 2 * \ral) / 11 * \i}, {\rbl - 0.01});
}

\node [above left = 1pt] at (H) {H};
\filldraw (H) circle[radius=0.2pt];
\node [above right = 1pt] at (I) {I};
\filldraw (I) circle[radius=0.2pt];
\node [below right = 1pt] at (J) {J};
\filldraw (J) circle[radius=0.2pt];
\node [below left = 1pt] at (K) {K};
\filldraw (K) circle[radius=0.2pt];

\draw [decorate,decoration={brace,amplitude=5pt,raise=0.5ex}] (J) -- (K) node[midway,yshift=-1.5em]{$1 - 2 a_{\lambda}$};
\draw [decorate,decoration={brace,amplitude=5pt,raise=0.5ex}] (I) -- (J) node[midway,xshift=1.5em]{$b_{\lambda}$};

\end{tikzpicture}
\caption{A visualization of the optimal transport solution for the partitioned pentagon-shaped and rectangle regions in the two-item case, for the case where the initial belief of the mechanism designer is $c=(\frac{1}{2}, \frac{1}{2})$ and the belief is updated by linear interpolating after one trade so that $\pi(c, x) = \lambda c + (1 - \lambda) x$, $\lambda \in [0, 1]$. The exact shapes demonstrated in this figure are for the case of $\lambda = 0.3$.}
\label{fig:linear_update_pentagon_and_rectangle}
\end{figure}

Fig.~\ref{fig:linear_update_partition} shows the optimal partition in this case with a linear belief update. In the optimal transport solution, the $+1$ point mass at $c=(\frac{1}{2}, \frac{1}{2})$, according to our definition of $\succeq$, can be spread outward (away from $c$) to uniformly cover the octagonal region at the center, at zero cost. In addition to the octagon, there are four equivalent rectangles and four equivalent pentagons. Each of the rectangles and pentagons contains a part of the four outer edges of the unit square, and the positive mass on that part of the edge gets distributed uniformly over the rectangle or pentagon to cancel out the negative mass. The total transport cost is therefore
\begin{equation*}
    C_{\lambda} = 4 C_{\lambda}^{(R)} + 4 C_{\lambda}^{(P)},
\end{equation*}
where $C_{\lambda}^{(R)}$ is the transport cost of one rectangle, $C_{\lambda}^{(P)}$ is the transport cost of one pentagon, and the subscript $\lambda$ represents the linear update parameter.

\paragraph{Transport plan} We parameterize the partition in terms of lengths $a_{\lambda}$ and $b_{\lambda}$. Fig.~\ref{fig:linear_update_pentagon_and_rectangle}, right, shows the optimal transport for one of the four equivalent rectangles. There is a positive on edge JK, and it is moved upward to uniformly cover the rectangle HIJK. The transport cost of one rectangle is therefore
\begin{equation*}
C_{\lambda}^{(R)} = \int_{0}^{b_{\lambda}} \frac{\lambda}{2} (1-2a_{\lambda}) \frac{1}{b_{\lambda}} x \mathrm{d}x = \frac{\lambda(1-2a_{\lambda})b_{\lambda}}{4}.
\end{equation*}
Fig.~\ref{fig:linear_update_pentagon_and_rectangle}, left, shows the optimal transport for one of the four equivalent pentagons. There is a positive mass on edges AB and BD. The positive mass at edge AB is moved downward to uniformly cover rectangle ABCF, the positive mass at edge CD is moved leftward to uniformly cover rectangle CDEG, and the positive mass at edge BC is moved to cover uniform triangle EFG. The transport cost for rectangle ABCF is
\begin{equation*}
    \int_{0}^{b_{\lambda}} \left(\frac{\lambda a_{\lambda}}{2}\right) \frac{1}{b_{\lambda}} x \mathrm{d}x = \frac{\lambda a_{\lambda} b_{\lambda}}{4}.
\end{equation*}
The transport cost for rectangle CDEG is
\begin{equation*}
    \int_{0}^{b_{\lambda}} \frac{\lambda}{2} (a_{\lambda} - b_{\lambda}) \frac{1}{b_{\lambda}} x \mathrm{d}x = \frac{\lambda (a_{\lambda} - b_{\lambda}) b_{\lambda}}{4}.
\end{equation*}
For the triangle EFG, the positive mass at edge BC is first moved to point C, then moved to point G, and finally moved to uniformly cover triangle EFG. This is the optimal way to transport the mass as the cost is with respect to the Manhattan distance of movements. Adding the costs of these three transport steps together, the transport for triangle EFG is
\begin{equation*}
    \left(\int_{0}^{b_{\lambda}} \frac{\lambda}{2} x \mathrm{d}x \right) + \left(\frac{\lambda b_{\lambda}}{2} b_{\lambda}\right) + \left(\int_{0}^{a_{\lambda} - b_{\lambda}} \int_{0}^{a_{\lambda}-b_{\lambda}-x} \left(\frac{\lambda b_{\lambda}}{2}\right) \left(\frac{2}{(a_{\lambda} - b_{\lambda})^2}\right)  (x+y) \mathrm{d}y \mathrm{d}x\right) = \frac{3 \lambda b_{\lambda}^2}{4} + \frac{\lambda (a_{\lambda} - b_{\lambda}) b_{\lambda}}{3}.
\end{equation*}
The total transport cost of the pentagon is 
\begin{equation*}
    C_{\lambda}^{(P)} = \frac{\lambda a_{\lambda} b_{\lambda}}{4} + \frac{\lambda (a_{\lambda} - b_{\lambda}) b_{\lambda}}{4} + \frac{3 \lambda b_{\lambda}^2}{4} + \frac{\lambda (a_{\lambda} - b_{\lambda}) b_{\lambda}}{3} = \frac{\lambda (5a_{\lambda} + b_{\lambda}) b_{\lambda}}{6}.
\end{equation*}

\paragraph{Solving for $a_{\lambda}$ and $b_{\lambda}$}
The structure of the transport plan constrains the values of $a_{\lambda}$ and $b_{\lambda}$.
There is a positive mass on edge JK of length $1-2a_{\lambda}$, and as the mass density per unit length of the edge is $\frac{\lambda}{2}$, the mass is $\frac{\lambda (1-2a_{\lambda})}{2}$. This mass is moved upward to be distributed uniformly over the rectangle HIJK and should perfectly cancel out the negative mass in the rectangle. 
The density of the initial negative mass over the rectangle HIJK is $-(2\lambda + 1) (1-2 a_{\lambda}) b_{\lambda}$, and because it should perfectly cancel out the positive mass $\frac{\lambda (1-2a_{\lambda})}{2}$, we have
\begin{align}
    \frac{\lambda (1-2a_{\lambda})}{2} -(2\lambda + 1) (1-2 a_{\lambda}) b_{\lambda} = 0.
\label{eq:linearupdate_ab_eq1}
\end{align}
The positive mass on edge BC is $\frac{\lambda b_\lambda}{2}$, and the negative mass on triangle EFG is $-\frac{(2\lambda + 1)(a_{\lambda}-b_{\lambda})^2}{2}$. As these two masses should perfectly cancel out each other, we have
\begin{equation}
    \frac{\lambda b_\lambda}{2} - \frac{(2\lambda + 1)(a_{\lambda}-b_{\lambda})^2}{2} = 0.
\label{eq:linearupdate_ab_eq2}
\end{equation}
Similarly, the positive mass on edge AB is $\frac{\lambda a_{\lambda}}{2}$, and the negative mass on rectangle ABCF is $-(2\lambda + 1)a_{\lambda}b_{\lambda}$. As these two masses should perfectly cancel out, we have
\begin{equation}
    \frac{\lambda a_{\lambda}}{2} - (2\lambda + 1)a_{\lambda}b_{\lambda} = 0.
\label{eq:linearupdate_ab_eq3}
\end{equation}
Combining equations Eq.~\ref{eq:linearupdate_ab_eq1}, Eq.~\ref{eq:linearupdate_ab_eq2}, and Eq.~\ref{eq:linearupdate_ab_eq3}, the values of $a_{\lambda}$ and $b_{\lambda}$ can be solved:
\begin{equation*}
    a_{\lambda} = \frac{(1 + \sqrt{2}) \lambda}{4 \lambda + 2}, \quad b_{\lambda} = \frac{\lambda }{4 \lambda + 2}.
\end{equation*}

\begin{figure}
    \centering
    \includegraphics[width=0.9\textwidth]{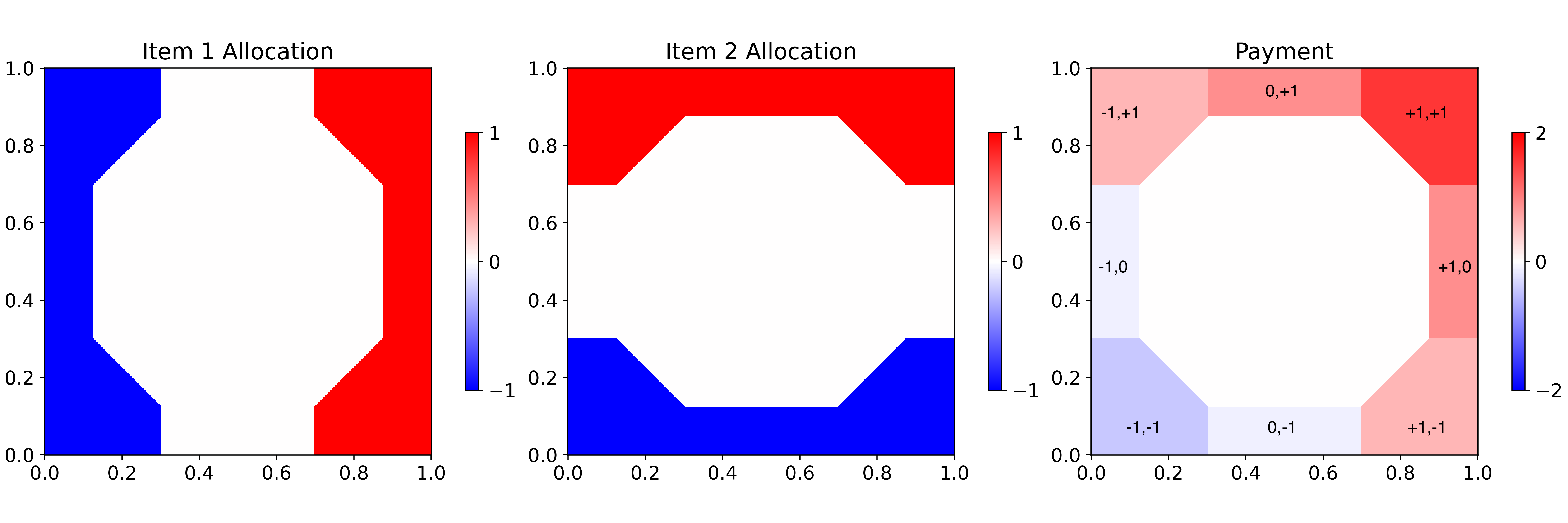}
    \caption{A plot of the allocation and payment rules for the true optimal mechanism under the linear update model ($\pi(c, x) = \lambda c + (1 - \lambda) x$) with $\lambda = \frac{1}{2}$ and $c=(\frac{1}{2}, \frac{1}{2}).$ Each distinct region is associated with a specific menu item; these menu items are marked on the payment rule plots.
    \label{fig:optimallinearupdatemec}}
\end{figure}

\paragraph{Computing the menu}
The boundary between two regions is the point where a trader is indifferent between the two menu items. At any trader type on a boundary, we can easily calculate the welfare of the two neighboring allocations, and infer the price that would make the trader indifferent.
The resulting menu is given in Table~\ref{tab:noisetradinglambda} and visualized (with a specific value of $\lambda = \frac{1}{2}$) in Fig.~\ref{fig:optimallinearupdatemec}. 
The total transport cost for this case with linear belief update is
\begin{equation}
    C_{\lambda} = 4 C_{\lambda}^{(R)} + 4 C_{\lambda}^{(P)} = 4 \cdot \frac{\lambda (1-2a_{\lambda})b_{\lambda}}{4} + 4 \cdot \frac{\lambda (5a_{\lambda} + b_{\lambda}) b_{\lambda}}{6} = \frac{\lambda^2 \left(\left(9+2 \sqrt{2}\right) \lambda +3\right)}{6 (2 \lambda +1)^2}.
\end{equation}
which exactly matches the expected profit of the optimal menu, so our solutions are optimal.

\section{Noise trading in 2 dimensions with an off-center belief}
\label{app:offcenterbelief}

We now consider another pure noise trading case which corresponds to the setting $\lambda = 1$ in the linear update model. Unlike in \S\ref{app:2dnoise}, the mechanism designer in this case has an off-center belief of $c = (\frac{1}{3}, \frac{1}{3})$.

\begin{figure}
    \centering
    \includegraphics[width=0.3\textwidth]{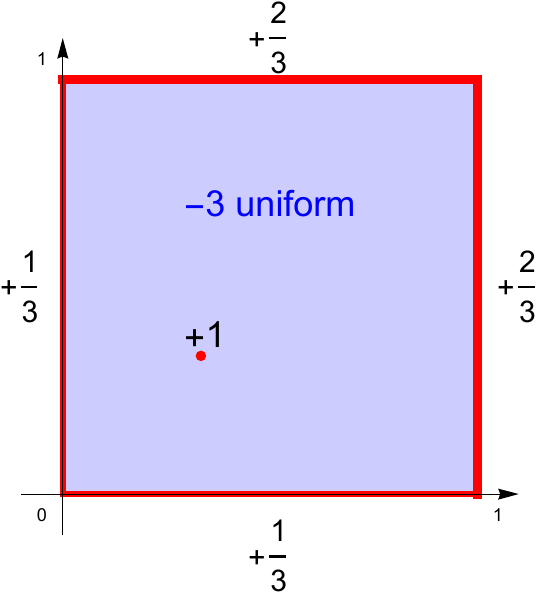}
    \caption{The transformed measure for the two good, adverse selection case on a uniform distribution, for $c=(\frac{1}{3}, \frac{1}{3})$. %
    }
    \label{fig:offcentermeasure}
\end{figure}

\subsection{Constructing the transformed measure}\label{app:offcentermeasure}
The transformed measure (Eq.~\ref{eq:mechobjectiverepeat}) for this case is:
\begin{itemize}
    \item $-(2\lambda + 1) = -3$ mass distributed uniformly.
    \item A $+1$ point mass at $c = (\frac{1}{3}, \frac{1}{3})$.
    \item Mass of $+(1-c_1)\lambda = +\frac{2}{3}$ on the right boundary, $+(1-c_2)\lambda = +\frac{2}{3}$ on the top boundary, $+ c_1\lambda = +\frac{1}{3}$ on the left boundary, and $+c_2\lambda = +\frac{1}{3}$ on the bottom boundary.
\end{itemize}
The positive and negative masses each have magnitude $3$. The measure is visualized in Fig~\ref{fig:offcentermeasure}.

\subsection{Conjecture for optimal mechanism}
We first train RochetNet to maximize profit in this case. The allocation rules for the mechanism after training are shown in Fig. \ref{fig:onethirdrochet}.

Similar to \S\ref{app:2dnoise} and \S\ref{app:linearbelief}, we use the RochetNet solution, as well as our knowledge of the problem setting, to conjecture the structure of the optimal mechanism. The optimal mechanism should be deterministic and should be symmetric with respect to the two items. The RochetNet solution suggests that the optimal mechanism offers 6 menu items in addition to the no-trade option.

\begin{figure}
\centering
\begin{tikzpicture}[scale=5]

\coordinate (V) at ({(1 / 3)}, {(1 / 3)});

\coordinate (A) at (0,1);
\coordinate (B) at (1,1);
\coordinate (C) at (1,0);
\coordinate (D) at (0,0);

\coordinate (E) at ({\rdo}, {1 - \reo});
\coordinate (F) at ({1 - \reo}, {\rdo});
\coordinate (G) at ({1 - \rfo}, {\rbo});
\coordinate (H) at ({\rao}, {\rbo});
\coordinate (I) at ({\rbo}, {\rao});
\coordinate (J) at ({\rbo}, {1 - \rfo});

\coordinate (K) at ({\rdo}, {1});
\coordinate (L) at ({1}, {\rdo});
\coordinate (M) at ({1 - \rfo}, {0});
\coordinate (N) at ({\rao}, {0});
\coordinate (O) at ({0}, {\rao});
\coordinate (P) at ({0}, {1 - \rfo});
\coordinate (II) at ({({1/63})*(33 - 18*sqrt(2) + 4*sqrt(74+22*sqrt(2)))}, {1/3});

\draw (A) -- (B) -- (C) -- (D) -- cycle;
\draw [line width=1pt, gray] (E) -- (F) -- (G) -- (H) -- (I) -- (J) -- cycle;

\draw [line width=1pt, blue]
    (M) -- (G)
    (N) -- (H)
    (O) -- (I)
    (P) -- (J);

\draw [line width=1pt, red]
    (K) -- (E)
    (L) -- (F);

\draw [line width=1pt, ao(english)]
    (D) -- (N)
    (D) -- (O);

\draw [line width=1pt, orange]
    (N) -- (M)
    (O) -- (P)
    (H) -- (G)
    (I) -- (J);

\draw [line width=1pt, brown]
    (L) -- (C)
    (K) -- (A);

\draw [line width=1pt, violet]
    (B) -- (L)
    (B) -- (K);

\draw [line width=1pt, green]
    (M) -- (C)
    (P) -- (A);

\draw [line width=1pt, cyan]
    (J) -- (E)
    (G) -- (F);

\draw [line width=1pt, yellow]
    (E) -- (F);

\node [above left = 0pt] at (A) {$(0, 1)$};
\node [above right = 0pt] at (B) {$(1, 1)$};
\node [below right = 0pt] at (C) {$(1, 0)$};
\node [below left = 0pt] at (D) {$(0, 0)$};
\node [above = 0pt] at (V) {$c=(\frac{1}{3}, \frac{1}{3})$};
\filldraw (A) circle[radius=0.3pt];
\filldraw (B) circle[radius=0.3pt];
\filldraw (C) circle[radius=0.3pt];
\filldraw (D) circle[radius=0.3pt];
\filldraw (V) circle[radius=0.3pt];

\node [left = 2pt] at (F) {$X$};
\node [above = 2pt] at (G) {$W$};
\filldraw (F) circle[radius=0.3pt];
\filldraw (G) circle[radius=0.3pt];

\draw [decorate,decoration={brace,amplitude=5pt,mirror,raise=0ex}] (D) -- (N) node[midway,yshift=-1em]{$a$};
\draw [decorate,decoration={brace,amplitude=5pt,mirror,raise=0ex}] (N) -- (M) node[midway,yshift=-1em]{$1 - a - f$};
\draw [decorate,decoration={brace,amplitude=5pt,mirror,raise=0ex}] (H) -- (N) node[midway,xshift=-1em]{$b$};
\draw [decorate,decoration={brace,amplitude=5pt,mirror,raise=0ex}] (C) -- (L) node[midway,xshift=1em]{$d$};
\draw [decorate,decoration={brace,amplitude=5pt,mirror,raise=0ex}] (L) -- (B) node[midway,xshift=1.6em]{$1 - d$};
\draw [decorate,decoration={brace,amplitude=5pt,mirror,raise=0ex}] (F) -- (L) node[midway,yshift=-1em]{$e$};
\draw [decorate,decoration={brace,amplitude=5pt,mirror,raise=0ex}] (M) -- (C) node[midway,yshift=-1em]{$f$};

\end{tikzpicture}
\caption{Partition of the unit square that corresponds to the optimal transport solution to the dual problem in the two-item pure noise trading case where the initial belief of the mechanism designer is $c=(\frac{1}{3}, \frac{1}{3})$. Edges of the same length are marked with the same color.}
\label{fig:off_center_partition}
\end{figure}
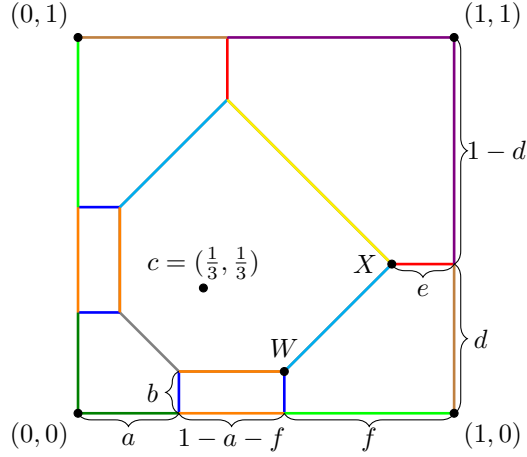

\paragraph{Partitioning the type space} The dual problem is the optimal transport according to the transformed measure in \S\ref{app:offcentermeasure} and the optimal transport plan involves partitioning the type space into separate regions. Based on the RochetNet solution, we conjecture the partitioned regions in Fig.~\ref{fig:off_center_partition}, which consist of a central hexagonal no-trade region, two rectangular regions, and four pentagonal regions.

\begin{figure}[t]
\centering
\begin{tikzpicture}[scale=10]

\coordinate (A) at ({0}, {\rao});
\coordinate (B) at ({\rao}, {\rao});
\coordinate (C) at ({\rao}, {\rao - \rbo});
\coordinate (D) at ({\rao}, {0});
\coordinate (E) at ({\rao - \rbo}, {0});
\coordinate (F) at ({0}, {\rao - \rbo});
\coordinate (G) at ({\rao - \rbo}, {\rao - \rbo});

\draw [line width = 1pt] (A) -- (B) -- (D) -- (E) -- (F) -- cycle;
\draw [line width = 1pt, dotted] (F) -- (C);
\draw [line width = 1pt, dotted] (G) -- (E);

\node [above left = 1pt] at (A) {A};
\filldraw (A) circle[radius=0.2pt];
\node [above right = 1pt] at (B) {B};
\filldraw (B) circle[radius=0.2pt];
\node [above right = 1pt] at (C) {C};
\filldraw (C) circle[radius=0.2pt];
\node [below right = 1pt] at (D) {D};
\filldraw (D) circle[radius=0.2pt];
\node [below left = 1pt] at (E) {E};
\filldraw (E) circle[radius=0.2pt];
\node [below left = 1pt] at (F) {F};
\filldraw (F) circle[radius=0.2pt];
\node [above = 1pt, fill=white] at (G) {G};
\filldraw (G) circle[radius=0.2pt];

\draw [decorate,decoration={brace,amplitude=5pt,raise=0.5ex}] (A) -- (B) node[midway,yshift=1.5em]{$a$};
\draw [decorate,decoration={brace,amplitude=5pt,mirror,raise=0.5ex}] (A) -- (F) node[midway,xshift=-1.5em]{$b$};
\draw [decorate,decoration={brace,amplitude=5pt,raise=0.5ex}] (C) -- (D) node[midway,xshift=2.5em]{$a - b$};
\draw [decorate,decoration={brace,amplitude=5pt,mirror,raise=0.5ex}] (E) -- (D) node[midway,yshift=-1.5em]{$b$};

\end{tikzpicture}
\begin{tikzpicture}[scale=10]

\coordinate (H) at ({0}, {\rbo});
\coordinate (I) at ({1 - \rao - \rfo}, {\rbo});
\coordinate (J) at ({1 - \rao - \rfo}, {0});
\coordinate (K) at ({0}, {0});

\draw [line width = 1pt] (H) -- (I) -- (J) -- (K) -- cycle;

\node [above left = 1pt] at (H) {H};
\filldraw (H) circle[radius=0.2pt];
\node [above right = 1pt] at (I) {I};
\filldraw (I) circle[radius=0.2pt];
\node [below right = 1pt] at (J) {J};
\filldraw (J) circle[radius=0.2pt];
\node [below left = 1pt] at (K) {K};
\filldraw (K) circle[radius=0.2pt];

\draw [decorate,decoration={brace,amplitude=5pt,raise=0.5ex}] (J) -- (K) node[midway,yshift=-1.5em]{$1 - a - f$};
\draw [decorate,decoration={brace,amplitude=5pt,raise=0.5ex}] (I) -- (J) node[midway,xshift=1.5em]{$b$};

\end{tikzpicture}
\begin{tikzpicture}[scale=10]

\coordinate (L) at ({0}, {1 - \rdo});
\coordinate (M) at ({1 - \rdo}, {1 - \rdo});
\coordinate (N) at ({1 - \rdo}, {1 - \rdo - \reo});
\coordinate (O) at ({1 - \rdo}, {0});
\coordinate (P) at ({1 - \rdo - \reo}, {0});
\coordinate (Q) at ({0}, {1 - \rdo - \reo});
\coordinate (R) at ({1 - \rdo - \reo}, {1 - \rdo - \reo});

\draw [line width = 1pt] (L) -- (M) -- (O) -- (P) -- (Q) -- cycle;
\draw [line width = 1pt, dotted] (Q) -- (N);
\draw [line width = 1pt, dotted] (R) -- (P);

\node [above left = 1pt] at (L) {L};
\filldraw (L) circle[radius=0.2pt];
\node [above right = 1pt] at (M) {M};
\filldraw (M) circle[radius=0.2pt];
\node [above right = 1pt] at (N) {N};
\filldraw (N) circle[radius=0.2pt];
\node [below right = 1pt] at (O) {O};
\filldraw (O) circle[radius=0.2pt];
\node [below left = 1pt] at (P) {P};
\filldraw (P) circle[radius=0.2pt];
\node [below left = 1pt] at (Q) {Q};
\filldraw (Q) circle[radius=0.2pt];
\node [above = 1pt, fill=white] at (R) {R};
\filldraw (R) circle[radius=0.2pt];

\draw [decorate,decoration={brace,amplitude=5pt,raise=0.5ex}] (L) -- (M) node[midway,yshift=1.5em]{$1 - d$};
\draw [decorate,decoration={brace,amplitude=5pt,mirror,raise=0.5ex}] (L) -- (Q) node[midway,xshift=-1.5em]{$e$};
\draw [decorate,decoration={brace,amplitude=5pt,raise=0.5ex}] (N) -- (O) node[midway,xshift=2.5em]{$1 - d - e$};
\draw [decorate,decoration={brace,amplitude=5pt,mirror,raise=0.5ex}] (P) -- (O) node[midway,yshift=-1.5em]{$e$};

\end{tikzpicture}
\begin{tikzpicture}[scale=10]

\coordinate (S) at ({0}, {\rfo});
\coordinate (T) at ({\rdo}, {\rfo});
\coordinate (U) at ({\rdo}, {\rfo - \reo});
\coordinate (V) at ({\rdo}, {0});
\coordinate (W) at ({\rdo - \rbo}, {0});
\coordinate (X) at ({0}, {\rfo - \reo});
\coordinate (Y) at ({\rdo - \rbo}, {\rfo - \reo});

\draw [line width = 1pt] (S) -- (T) -- (V) -- (W) -- (X) -- cycle;
\draw [line width = 1pt, dotted] (X) -- (U);
\draw [line width = 1pt, dotted] (Y) -- (W);

\node [above left = 1pt] at (S) {S};
\filldraw (S) circle[radius=0.2pt];
\node [above right = 1pt] at (T) {T};
\filldraw (T) circle[radius=0.2pt];
\node [above right = 1pt] at (U) {U};
\filldraw (U) circle[radius=0.2pt];
\node [below right = 1pt] at (V) {V};
\filldraw (V) circle[radius=0.2pt];
\node [below left = 1pt] at (W) {W};
\filldraw (W) circle[radius=0.2pt];
\node [below left = 1pt] at (X) {X};
\filldraw (X) circle[radius=0.2pt];
\node [above = 1pt, fill=white] at (Y) {Y};
\filldraw (Y) circle[radius=0.2pt];

\draw [decorate,decoration={brace,amplitude=5pt,raise=0.5ex}] (S) -- (T) node[midway,yshift=1.5em]{$d$};
\draw [decorate,decoration={brace,amplitude=5pt,mirror,raise=0.5ex}] (S) -- (X) node[midway,xshift=-1.5em]{$e$};
\draw [decorate,decoration={brace,amplitude=5pt,raise=0.5ex}] (U) -- (V) node[midway,xshift=2.5em]{$f - e$};
\draw [decorate,decoration={brace,amplitude=5pt,mirror,raise=0.5ex}] (W) -- (V) node[midway,yshift=-1.5em]{$b$};

\end{tikzpicture}

\caption{A visualization of the the partitioned pentagonal and rectangular regions in the conjectured optimal mechanism for the two-item pure noise trading case where the belief of the mechanism designer is $c=(\frac{1}{3}, \frac{1}{3})$.}
\label{fig:off_center_pentagon_and_rectangle}
\end{figure}
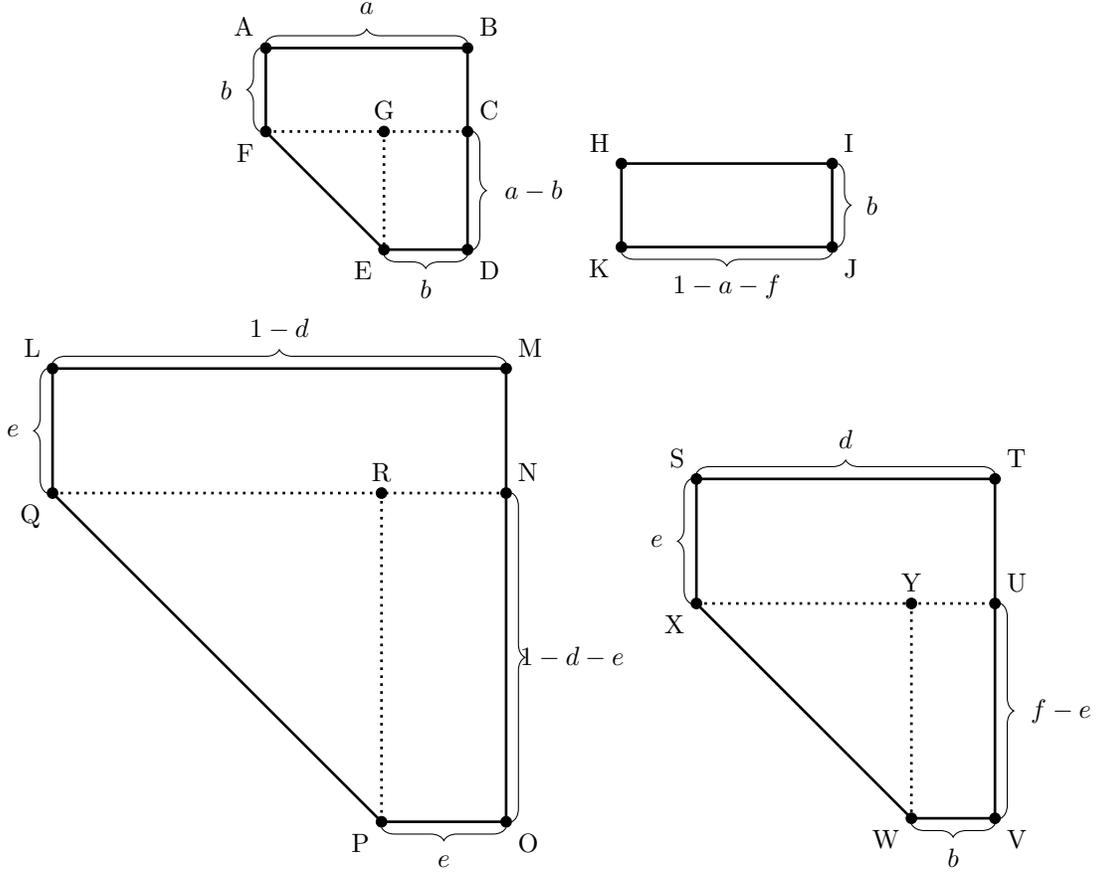

In the optimal transport solution under our conjecture, the $+1$ point mass at $c=(\frac{1}{3}, \frac{1}{3})$, according to our definition of $\succeq$, can be spread outward (away from $c$) to uniformly cover the hexagonal region at the center, at zero cost. In addition to the hexagon, there are two equivalent rectangular regions, and we denote this shape with $R$. There are two equivalent pentagonal regions at top left and bottom right, which we denote with $P_1$. We denote the smaller pentagonal region at bottom left with $P_2$, and the larger pentagonal region at top right with $P_3$. Each of the rectangles and pentagons contains a part of the four outer edges of the unit square, and the positive mass on that part of the edge gets distributed uniformly over the rectangle or pentagon to cancel out the negative mass.

\paragraph{Solving for $a$, $b$, $d$, $e$, and $f$}
The structure of the transport plan constrains the values of $a$, $b$, $d$, $e$, and $f$.
For the rectangular shape $R$ illustrated with rectangle HIJK, there is a positive mass of $\frac{1 - a - f}{3}$ on edge JK. This mass is moved upward to be distributed uniformly over the rectangle HIJK and should perfectly cancel out the negative mass of $-3b(1-a-f)$ in the rectangle. Therefore, we have
\begin{align}
    \frac{1 - a - f}{3} - 3b(1-a-f) = 0.
\label{eq:off_center_ab_eq1}
\end{align}
For the pentagonal shape $P_1$ illustrated with pentagon STVWX, there is a positive mass of $\frac{2d}{3}$ on edge ST and a positive mass of $\frac{f}{3}$ on edge TV. They should cancel out the negative mass of $-3\left(df - \frac{(d-b)(f-e)}{2}\right)$ perfectly. That is,
\begin{align}
    \frac{2d}{3} + \frac{f}{3} - 3\left(df - \frac{(d-b)(f-e)}{2}\right) = 0.
\label{eq:off_center_ab_eq2}
\end{align}
Similarly, for the pentagonal shape $P_2$, it should hold that
\begin{align}
    \frac{a}{3} + \frac{a}{3} - 3\left(a^2 - \frac{(a-b)^2}{2}\right) = 0.
\label{eq:off_center_ab_eq3}
\end{align}
For the pentagonal shape $P_3$, we have
\begin{align}
    \frac{2(1-d)}{3} + \frac{2(1-d)}{3} - 3\left((1-d)^2 - \frac{(1-d-e)^2}{2}\right) = 0.
\label{eq:off_center_ab_eq4}
\end{align}
Consider the point $X = (1 - e, d)$ in Fig.~\ref{fig:off_center_partition}. The fact that this point lies on the boundary between the central no-trade region and the bottom right pentagonal region implies that a trader with the type $(1 - e, d)$ is indifferent to whether making no trade at all, or getting the allocation of $(+1, -1)$ and making the corresponding payment. Assume the payment for the allocation $(+1, -1)$ is $p$. Therefore, it holds that
\begin{align}
    (+1) \cdot (1 - e) + (-1) \cdot d - p = 0.
\label{eq:off_center_point_x}
\end{align}
Similarly, we can get the following equation by considering the point $W = (1 - f, b)$ in Fig.~\ref{fig:off_center_partition}:
\begin{align}
    (+1) \cdot (1 - f) + (-1) \cdot b - p = 0.
\label{eq:off_center_point_w}
\end{align}
Additionally, we have the following constraints
\begin{align}
    0 \le a, b, d, e, f \le 1, \quad e < f, \quad 1 - a - f > 0.
\label{eq:off_center_constraints}
\end{align}
Combining equations Eq.~\ref{eq:off_center_ab_eq1}, Eq.~\ref{eq:off_center_ab_eq2}, Eq.~\ref{eq:off_center_ab_eq3}, Eq.~\ref{eq:off_center_ab_eq4}, Eq.~\ref{eq:off_center_point_x}, Eq.~\ref{eq:off_center_point_w}, and Eq.~\ref{eq:off_center_constraints}, the values of $a$, $b$, $d$, $e$, and $f$ can be solved:
\begin{align}
    & a = \frac{1 + \sqrt{2}}{9}, \quad b = \frac{1}{9}, \quad d = \frac{1}{63} \left(25-6 \sqrt{2}+2 \sqrt{134-82 \sqrt{2}}\right), \\
    & e = \frac{1}{63} \left(26+24 \sqrt{2}-2 \sqrt{310+214 \sqrt{2}}\right), \quad f = \frac{1}{63} \left(44+18 \sqrt{2}-4 \sqrt{74+22 \sqrt{2}}\right).
\end{align}

\begin{figure}
    \centering
    \includegraphics[width=0.9\textwidth]{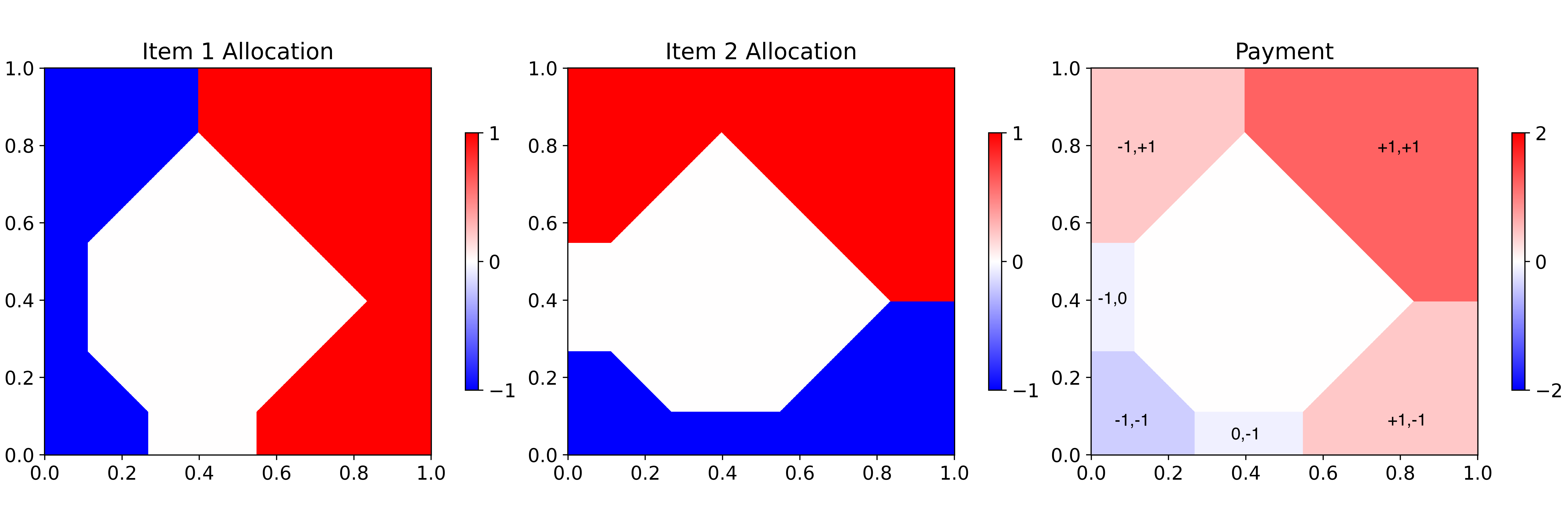}
    \caption{A plot of the allocation and payment rules for the conjectured optimal mechanism under the noise trading model ($\pi(c, x) = c$) for $c=(\frac{1}{3}, \frac{1}{3}).$ Each distinct region is associated with a specific menu item; these menu items are marked on the payment rule plots.
    \label{fig:conjectured_mec_off_center}}
\end{figure}

\paragraph{Computing the menu}
The boundary between two regions is the point where a trader is indifferent between the two menu items. At any trader type on a boundary, we can easily calculate the welfare of any two neighboring allocations, and infer the price that would make the trader indifferent between them.
The resulting menu is given in Table~\ref{tab:noisetradingonethird} and visualized in Fig.~\ref{fig:conjectured_mec_off_center}. 

\end{document}